\documentclass[runningheads]{llncs}

\usepackage{microtype}
\usepackage{graphicx}
\usepackage{caption}
\usepackage{subcaption}
\usepackage{float}
\usepackage{multirow}
\usepackage[ruled,linesnumbered]{algorithm2e}
\usepackage{mathtools}
\usepackage{tikz} 
\usetikzlibrary{calc}
\usepackage{hyperref}
\usepackage[capitalise]{cleveref}

\usepackage{amsfonts,amsmath,amssymb}
\usepackage{soul}
\usepackage{xcolor}
\usepackage{xspace}
\usepackage{paralist}
\usepackage{complexity}
\usepackage{comment}
\usepackage{graphics}
\usepackage{adjustbox}

\definecolor{mygreen}{rgb}{0, 0.5, 0}
\newcommand{\codecomment}[1]{{{\textcolor{mygreen}{#1}\xspace}}}

\Crefname{observation}{Observation}{Observations}
\Crefname{algorithm}{Algorithm}{Algorithms}
\Crefname{section}{Sect.}{Sects.}
\Crefname{observation}{Observation}{Observations}
\Crefname{lemma}{Lemma}{Lemmas}
\Crefname{claim}{Claim}{Claims}
\Crefname{figure}{Fig.}{Figs.}
\Crefname{figure}{Fig.}{Figs.}
\Crefname{invariant}{Inv.}{Invs.}
\Crefname{enumi}{Condition}{Conditions}
\Crefname{property}{Property}{Properties}
\Crefname{assumption}{Assumption}{Assumptions}

\newcommand{\blue}[1]{{{\textcolor{blue}{#1}\xspace}}}

\renewcommand{\emph}[1]{\blue{\em #1}}
\SetKwInput{KwInput}{Input}
\SetKwInput{KwOutput}{Output}
\SetNlSty{}{}{}

\SetCommentSty{mycommfont}
\let\oldnl\nl
\newcommand\nonl{
\renewcommand{\nl}{\let\nl\oldnl}}

\hypersetup{colorlinks=true,urlcolor=blue,citecolor=blue,linkcolor=blue}

\newcommand{\qedclaim}{\hfill $\blacksquare$}

\begin{document}
\title{Treebar Maps: Schematic Representation\\ of Networks at Scale}
\titlerunning{Treebar Maps}
\author{Giuseppe Di Battista\texorpdfstring{\href{https://orcid.org/0000-0003-4224-1550}{\protect\includegraphics[scale=0.45]{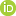}}}{} \and
	Fabrizio Grosso\texorpdfstring{\href{https://orcid.org/0000-0002-5766-4567}{\protect\includegraphics[scale=0.45]{orcid}}}{}\and
	Silvia Montorselli\and \\
	Maurizio Patrignani\texorpdfstring{\href{https://orcid.org/0000-0001-9806-7411}{\protect\includegraphics[scale=0.45]{orcid}}}{}}
\authorrunning{Di Battista et al.}
\institute{Roma Tre University, Rome, Italy
	\email{\{giuseppe.dibattista,fabrizio.grosso,maurizio.patrignani\}@uniroma3.it, sil.montorselli@stud.uniroma3.it}}

\maketitle

\begin{abstract}
Many data sets, crucial for today's applications, consist essentially of enormous networks, containing millions or even billions of elements. Having the possibility of visualizing such networks is of paramount importance. We propose an algorithmic framework and a visual metaphor, dubbed treebar map, to provide schematic representations of huge networks. Our goal is to convey the main features of the network's inner structure in a straightforward, two-dimensional, one-page drawing. This drawing effectively captures the essential quantitative information about the network's main components. Our experiments show that we are able to create such representations in a few hundreds of seconds. We demonstrate the metaphor's efficacy through visual examination of extensive graphs, highlighting how their diverse structures are instantly comprehensible via their representations.
\end{abstract}

\section{Introduction}

Many data sets, crucial for today's applications, consist essentially of enormous networks, containing millions or even billions of elements. Hence, having the possibility of visualizing such networks is of paramount importance.
However, drawing all the details of such huge networks using node-link representations is obviously unpractical (see also \cite{DBLP:journals/tvcg/YoghourdjianYDL21}) even if the layout could be efficiently computed in a distributed environment~\cite{Arleo2017124}. 
This originated a research strain devoted to finding methods that produce drawings where either the graph is only partially represented or it is schematically visualized. 

Proxy drawings, such as those described in~\cite{DBLP:conf/apvis/NguyenMLE18,DBLP:journals/ans/HuCHCMTEM21,DBLP:journals/cg/ImreTWZFW20,DBLP:journals/tvcg/NguyenHEM17}, fall under the category of the first type. In a proxy drawing, a graph is depicted using a smaller graph drawing that retains certain characteristics of the original graph. 

Several types of schematic representations have been proposed. See, e.g., \cite{AbelloZNHA22,dfpt-srlbg-jgaa-21,DBLP:journals/tvcg/YoghourdjianDKM18,terrain}.
In~\cite{dfpt-srlbg-jgaa-21} algorithms are presented for constructing schematic representations of graphs that can be decomposed using a set of separating vertices.
In graph thumbnails~\cite{DBLP:journals/tvcg/YoghourdjianDKM18}, every connected component of a graph is depicted as a disk, while biconnected components are represented by disks nested within the disk of the connected component they belong to. Each biconnected component comprises a stack of disks, with each disk representing a set of vertices with the same ``coreness''. The size of each disk is proportional to the number of vertices in the set.
The approach of~\cite{terrain} uses a $3$-dimensional representation of a tree-map, where the tree is given by the inclusion relationships between sets of vertices having certain scalar attributes.
In \cite{AbelloZNHA22} the proposed metaphor for the schematic representation is inspired to the map of a city.
An approach that is somehow in between proxy drawings and schematic representation is proposed in~\cite{browvis-22}, where, for graphs with up to two million edges, groups of vertices belonging to the same cluster are represented by clouds and only inter-cloud edges are shown.

A vertex of a graph $G$ has \emph{coreness $k$} (see \cref{fig:example-node-link}) if it belongs to the maximal induced subgraph of $G$ such that each vertex has degree at least $k$, but it does not belong to the maximal induced subgraph of $G$ such that each vertex has degree at least $k+1$~\cite{seidman-83}. 
Intuitively, the higher the node coreness the more the node belongs to the denser innermost kernel of the graph. Each connected component of the set of vertices with coreness at least $k$ is called a \emph{$k$-core}.
Hence, a $k$-core is a maximal induced connected subgraph such that each node has degree at least $k$~\cite{seidman-83}. Equivalently, a $k$-core is a connected component of the graph obtained from $G$ by repeatedly deleting all vertices of degree less than~$k$~\cite{batagelj-03}.

$k$-cores have several applications in social networks, bioinformatics, and neurosciences. A survey can be found in \cite{DBLP:journals/vldb/MalliarosGPV20}.
They have been also used for the visualization of graphs within the node-link metaphor.
In \cite{DBLP:conf/nips/Alvarez-HamelinDBV05}, vertices in the same $k$-core are placed on the same circle of a sequence of concentric circles. However, the drawing standard does not seem suitable for very large graphs.
The authors of \cite{nguyen-17} propose a concentric circle placement of the vertices, based on their coreness, and a variation of a force-directed layout to effectively display the structure of the graphs. The technique is tested on graphs with at most $200,000$ edges.

In this paper we consider the problem of succinctly conveying the inner structure of the denser portions of massive graphs. To this aim, we represent the number of $k$-cores, their sizes, and their containment relationship (simultaneously for the whole range of values of $k$) with a two-dimensional, one-page diagram that we call \emph{treebar map}. The treebar map combines the ability of conveying inclusion relationships of a nested treemap with the effectiveness of representing scale-free quantities of a logarithmic bar-chart. In order to scale with respect to the number of coreness values of the vertices of huge graphs, we allow for a simplification of the treebar map analogous to the one that is used for contour lines showing elevations in a geographic map.  
In order to scale also with respect to computational times, we exploit a data structure that we call \emph{core-connectivity tree}. Such a tree is similar to the ``scale tree'' of~\cite{terrain} where the scalar attribute assigned to each vertex is its coreness. We describe an algorithm to produce core-connectivity trees in time $O(n + m\alpha(n))$, where $\alpha(n)$ is the inverse of the Ackermann function, meeting the efficiency bound that was envisaged in~\cite{terrain}.   

We show the effectiveness of treebar maps by presenting the first visual analysis of several real-world graphs ranging from 10 million to about 2 billion edges.
Due to the size of the considered instances, we could not compare with the alternative 3D visualization of~\cite{AbelloZNHA22,terrain}, which admittedly have scalability problems~\cite{terrain}, and which need full 3D navigation to avoid occlusion~\cite{AbelloZNHA22,terrain}.
On the other hand, Thumbnails~\cite{DBLP:journals/tvcg/YoghourdjianDKM18} does not represent the $k$-core of the network, but only how the number of vertices of a biconnected component change with the coreness value. Also, this metaphor has been tested on graphs with up to a couple of million of edges.

\section{Core-Connectivity Trees}\label{sec:core-connectivity-trees}

To efficiently build treebar maps we exploit a data-structure that we call ``core-connectivity tree''.
We only consider graphs that are undirected, without isolated vertices, and without self-loops.
 
We denote by $c(u)$ the coreness of a vertex $u$.
According to the definition, the $1$-cores of $G$ are its connected components.
Conventionally, we assume the $0$-core of~$G$ to contain all the vertices of~$G$.
Observe that if $G$ is connected its unique $1$-core and its $0$-core coincide.
We define $c(G)$ as the maximum~$k$ such that $G$ has a non-empty $k$-core.
Consider a $k$-core $G_k$ and a $(k+1)$-core $G_{k+1}$ of $G$. We have the following property.

\begin{property}[\cite{seidman-83,batagelj-03}]\label{pr:core-subsets}
Either $G_k$ and $G_{k+1}$ are vertex-disjoint or the set of vertices of $G_{k+1}$ is a subset of the set of vertices of $G_k$.
\end{property}

The \emph{core-connectivity tree} $T$ of an $n$-vertex graph $G$ is the rooted tree defined as follows. The leaves of $T$ are the vertices of $G$. The non-leaf nodes of $T$ are the $k$-cores of $G$, with $k=0, \dots, c(G)$. The root of $T$ is its $0$-core. Let $u$ be a leaf of $T$, its parent is the $k$-core with maximum $k$ containing $u$. A $k$-core $G_k$ is the parent of a $(k+1)$-core $G_{k+1}$ if the set of vertices of $G_{k+1}$ is a subset (not necessarily proper) of the set of vertices of $G_k$. The fact that $T$ is a tree directly descends from \cref{pr:core-subsets}. 
Observe that a $k$-core and a $(k+1)$-core can have the same set of vertices. Hence, $T$ can have nodes with just one child. To save space we contract the chains of nodes with one child into a single node. Namely, let $\mu_{k}, \dots, \mu_{k+h}$ be a maximal chain of nodes, corresponding to $k$-, $(k+1)$-, \dots, $(k+h)$-cores, respectively, and such that the only child of $\mu_i$ is $\mu_{i+1}$ $i=1,\dots,k+h-1$, we contract $\mu_{k}, \dots, \mu_{k+h}$ into one node $\mu$, corresponding to all the $i$-cores of the nodes of the chain. Hence, we denote by $\min_c (\mu)=k$ and by $\max_c (\mu)=k+h$ the minimum and maximum corenesses associated with~$\mu$, respectively. More generally, we label all nodes of $T$ with their minimum and maximum coreness.

\begin{algorithm}[h!]
\caption{{\bf Procedure} \textsc{InsertEdge}$(F,(u,v))$}
\label{alg:insert-edge}
\SetKwInOut{Input}{Input}
\SetKwInOut{Output}{Output}
\SetKwInOut{SubProcedures}{Uses}
\SetKwInOut{SideEffects}{Side Effects}
\Input{A forest $F$ whose leaves are the vertices of $G$, an edge $(u,v)$ of~$G$}
\SubProcedures{%
{\sc Leaf}$(F,u)$: Returns the leaf of $F$ corresponding to $u$ \\
{\sc Root}$(\mu)$: Returns the root of the tree containing $\mu$ in $F$ \\
{\sc NewNodeWithChildren}$(\mu_1,\mu_2)$: Creates a new node for $F$\\ ~~~~~and adds children $\mu_1$ and $\mu_2$\\
{\sc AddChild}$(\rho,\mu)$: Adds $\mu$ as a child of $\rho$ in $F$ \\
}
\SideEffects{Updates $F$}
\BlankLine
$\mu_u \leftarrow \textsc{Leaf}(F,u)$; \\
$\mu_v \leftarrow \textsc{Leaf}(F,v)$; \\
$\rho_u \leftarrow \textsc{Root}(\mu_u)$; \\
$\rho_v \leftarrow \textsc{Root}(\mu_v)$; \\
\uIf{$(\rho_u=\mu_u) \textrm{~and~} (\rho_v=\mu_v)$ }{
        $\rho \leftarrow \textsc{NewNodeWithChildren}(\mu_u,\mu_v)$; \\
        $\rho.maxcoreness \leftarrow c(u,v)$;
}
\uElseIf{$(\rho_u \neq \mu_u) \textrm{~and~} (\rho_v=\mu_v)$}{
    \lIf{$\rho_u.maxcoreness = c(u,v)$}{
       $\textsc{AddChild}(\rho_u,\mu_v)$
    }\uElse{
        \codecomment{/* Necessarily $\rho_u.maxcoreness > c(u, v)$ */}\\
        $\rho \leftarrow \textsc{NewNodeWithChildren}(\rho_u,\mu_v)$; \\
        $\rho.maxcoreness \leftarrow c(u,v)$;         
    }    
}
\uElseIf{$(\rho_u = \mu_u) \textrm{~and~} (\rho_v \neq \mu_v)$}{
    \lIf{$\rho_v.maxcoreness = c(u,v)$}{
       $\textsc{AddChild}(\rho_v,\mu_u)$
    }\uElse{
        \codecomment{/* Necessarily $\rho_v.maxcoreness > c(u, v)$ */}\\
        $\rho \leftarrow \textsc{NewNodeWithChildren}(\mu_u,\rho_v)$; \\
        $\rho.maxcoreness \leftarrow c(u,v)$;         
    }    
}

\lElseIf{$(\rho_u = \rho_v)$}{
    \Return \codecomment{~/* Nothing to do */}
}
\uElse{
    \codecomment{/* $(\rho_u \neq \mu_u) \textrm{~and~} (\rho_v \neq \mu_v) \textrm{~and~} (\rho_u \neq \rho_v)$ */} \\
    \lIf{$\rho_u.maxcoreness < \rho_v.maxcoreness$}{
        $\textsc{AddChild}(\rho_v,\rho_u)$
    }
    \uElse{
        \codecomment{/* $\rho_u.maxcoreness \geq \rho_v.maxcoreness$ */} \\
        $\textsc{AddChild}(\rho_u,\rho_v)$;        
    }
}

\end{algorithm}

\begin{theorem}\label{th:core-connectivity-tree}
Le $G$ be a graph with $n$ vertices and $m$ edges. The core-connectivity tree of $G$ has $O(n)$ nodes and can be computed in $O(n+m\alpha(n))$ time, where $\alpha(n)$ is the inverse of the Ackermann function.
\end{theorem}

\begin{proof}
According to the definition of core-connectivity tree we have that all the non-leaf nodes of $T$ have at least two children and, since the leaves of $T$ are the $n$ vertices of $G$, the number of nodes of $T$ is $O(n)$. 
First observe that the value of $\min_c(\mu)$ for a non-leaf node $\mu$ of $T$ can be easily obtained as $\max_c(\nu)+1$, where $\nu$ is the parent of $\mu$. Hence, in the following we only label each non-leaf node $\mu$ of $T$ with the value of $\max_c(\mu)$. The value of $\min_c(\mu)$ is computed at the end of the process.
Second, observe that if some non-leaf node $\mu$ of $T$ with maximum coreness $\max_c(\mu)$ and with children $\nu_1, \nu_2, ..., \nu_h$ is replaced with a subtree $T_\mu$ such that: (i) all internal nodes of $T_\mu$ are labeled with  $\max_c(\mu)$; (ii) all internal nodes of $T_\mu$ have degree at least two; and (iii) the leaves of $T_\mu$ are $\nu_1, \nu_2, ..., \nu_h$, then tree $T'$ would also have $O(n)$ nodes. 
Also, given $T'$ we can construct in linear time the core-connectivity tree $T$ by contracting all edges that join a pair of nodes $\mu_1$ and $\mu_2$ with $\max_c(\mu_1)=\max_c(\mu_2)$. 
Based on the above consideration we focus on constructing a tree~$T'$. We first describe an $O(nm)$-time algorithm for constructing $T'$ and then we refine it to an $O(n+m\alpha(n))$-time algorithm.

We label all vertices of $G$ with their coreness in $O(n+m)$ time with the algorithm in \cite{DBLP:journals/corr/cs-DS-0310049,DBLP:journals/jacm/MatulaB83}.
Then, we label each edge $(u,v)$ with its \emph{coreness} $c(u,v)$, defined as the minimum between $c(u)$ and $c(v)$. Observe that $c(u,v)$ is the higher value of $k$ for which $(u,v)$ belongs to a $k$-core. We sort the edges in decreasing order of their coreness. This can be done in $O(m)$ time by using a bucket-sort.

Finally, we launch an iterative procedure that, for each edge in decreasing order of coreness, updates a forest $F$ that at the end of the process yields a tree~$T'$. 
Forest $F$ is initialized wiht one isolated node $n(v)$ for each vertex $v$ of $G$. Each non-leaf node $\mu$ of $F$ will be labeled with a maximum coreness value $\max_c(\mu)$.  
We launch Procedure \textsc{InsertEdge}$(F,(u,v))$ on the current edge $(u,v)$ (refer to \cref{alg:insert-edge}). The procedure first finds the roots $\rho_u$ and $\rho_v$ of the two trees of $F$ containing the leaves $\mu_u$ and $\mu_v$ corresponding to $u$ and $v$, respectively (lines 1--4). We have four cases:

\emph{Case 1:} If $\rho_u = \mu_u$ and $\rho_v = \mu_v$ we create a parent $\rho$ for $\mu_u$ and $\mu_v$ and set $\max_c(\rho)= c(u,v)$ (lines 6--7). 
\emph{Case 2:} If $\rho_u \neq \mu_u$ and $\rho_v = \mu_v$ then we have two cases. If $c(u,v) = \max_c(\rho_u)$ then we add $\mu_v$ as a child of $\rho_u$. Otherwise (i.e., $c(u,v) < \max_c(\rho_u)$), we create a parent $\rho$ for $\rho_u$ and $\mu_v$ and set $\max_c(\rho)=c(u,v)$.
\emph{Case 3:} If $\rho_u = \rho_v$, then necessarily $c(u,v) = \max_c(\rho_u)$ and we proceed with the next edge.
\emph{Case 4:} If $\rho_u \neq \mu_u$, $\rho_v \neq mu_v$, and $\rho_u \neq \rho_v$ then 
we have two cases. If $c(\rho_u) < \max_c(\rho_v)$, then we set $\rho_v$ as a child of $\rho_v$ (observe that in this case $c(u,v) = \max_c(\rho_v)$). If $\max_c(\rho_u) = \max_c(\rho_v)$, then we set an arbitrary chosen root, say $\rho_v$, as a child of the other root, in this case $\rho_u$.

The above described algorithm has complexity $O(nm)$. Indeed, we have $m$ iterations, one for each edge $(u,v)$, and the search for the root of the trees $u$ and $v$ belong to may take $O(n)$ time, as the height of the trees of $F$ may be $O(n)$ (recall that two adjacent nodes may have the same maximum coreness).

In order to refine the above algorithm to an $O(n+m\alpha(n))$-time one, we associate $F$ with a Union-Find data structure~\cite{cormen01}. This data structure represents a collection of disjoint sets of elements and elects a representative element for each set. It supports the $O(1)$-time operation of union of two sets and the $O(\alpha(n))$-time search for the representative element of the set a given element belongs to. We use the Union-Find data-structure to allow us the efficient retrieval of the root of the tree of $F$ a node $n(v)$ belongs to. Precisely, for each tree $T_\rho$ of $F$ with root $\rho$ we associate a set $S_\rho$ containing the leaves of $T_\rho$ and we add a pointer from the representative element of $S_\rho$ to $\rho$. When an edge $(u,v)$ is processed, we find in $O(\alpha(n))$-time the representative elements of $n(u)$ and $n(v)$ and, therefore, the roots $\rho(n(u))$ and $\rho(n(v))$. When two trees with roots $\rho'$ and $\rho''$ of $F$ are joined, we merge the corresponding sets $S_{\rho'}$ and $S_{\rho''}$ and point from the representative element of the obtained set $S_\rho = S_{\rho'} \cup S_{\rho''}$ to the root $\rho$ of the obtained tree $T_\rho$ of $F$. 
Once a single tree $T'$ has been obtained, we construct in linear time the core-connectivity tree $T$ by contracting all edges that join a pair of nodes $\mu_1$ and $\mu_2$ with $\max_c(\mu_1)=\max_c(\mu_2)$. Finally, we compute the value of $\min_c(\mu)$ for each internal node $\mu$ of $T$. 
\end{proof}

\section{From Core-Connectivity Trees to Treebar Maps}

In the previous section, we discussed the efficient computation of a core-connecti\-vi\-ty tree, denoted as $T$, for a given graph $G$. Now, let's remove from $T$ the leaves (representing the vertices of $G$) and let's explore a possible method to construct a schematic representation of $G$ using a similar approach outlined in \cite{DBLP:journals/tvcg/YoghourdjianDKM18}. This approach involves creating a tree-map representation of $T$, where each $k$-core corresponds to a closed region whose area is proportional to the size of the $k$-core. Additionally, the inclusion relationship between a $k$-core $a$ and a $(k+1)$-core $b$ is represented by the inclusion between the regions representing $a$ and $b$. To enhance visual perception of the relationship between coreness and density, we can assign colors to the regions based on their coreness. For example, we can use a color scheme where the color of $b$ is darker than the color of $a$, indicating their respective coreness levels.

However, this simple plan encounters two scalability challenges that must be addressed in order to create a practical schematic representation of $G$. The first challenge, referred to as Challenge \emph{CS}, arises due to the varying cardinalities of the $k$-cores involved in $T$. Sets with significantly different magnitudes of cardinalities pose difficulties when attempting to represent them effectively within the tree-map structure. The second challenge, that we call Challenge \emph{CT}, is associated with the size of $T$. As $G$ increases in size or complexity, the corresponding tree $T$ also grows in size, and we will see in \cref{sec:experiments} that graphs with tens of millions of edges can have connectivity-trees with thousands of internal nodes. Visualizing large trees can be challenging in terms of the clarity of the representation.

When addressing Challenge CS, one potential solution would be to assign a non-linear area to the region representing a $k$-core $\mu$, with the aim of accommodating the varying cardinalities of the sets. For instance, a logarithmic dependency on the size of $\mu$ could be considered. However, implementing such a solution would conflict with the tree-map representation, specifically when dealing with inclusion relationships. To illustrate this, let's consider two $k+1$-cores, $\mu_1$ and $\mu_2$, both contained within a $k$-core $\mu$. It is possible that the logarithm of the size of $\mu_1$ added to the logarithm of the size of $\mu_2$ could be smaller than the logarithm of the size of $\mu$. This discrepancy can lead to misleading tree-map representations, particularly when comparing the sizes of the involved sets. Therefore, utilizing a logarithmic or non-linear area assignment, while attempting to address Challenge CS, would introduce inconsistencies in accurately representing the set sizes within the tree-map structure. 

When considering Challenge CT, one possible solution could be to focus on visualizing only the lower sections of $T$, that contain the $k$-cores with the largest $k$ values. This approach aims to highlight the most significant or interesting parts of the graph's structure while disregarding the rest. However, it is important to note that this approach contradicts the intention of providing a comprehensive and high-level overview of the entire graph~$G$. While focusing solely on the highest $k$-cores may be useful for specific analysis purposes, it may not fulfill the broader goal of offering a comprehensive visual summary of the entire graph.

To address the challenges of CS and CT, we introduce a novel visualization paradigm called \emph{treebar map}, which incorporates the following concepts (see in \cref{fig:example-node-link} the node-link drawing of the graph whose treebar map in represented in \cref{fig:example-1}). A treebar map combines a \emph{horizontal-treemap} $\cal H$ to depict the inclusion relationships between $k$-cores with a \emph{bar-chart} $\cal B$ to display the sizes of $k$-cores. The treemap $\cal H$ is horizontal so that each of its regions is topped by the bar of $\cal B$ representing its size in logarithmic scale.

\begin{figure}[tb!]
    \captionsetup[subfigure]{justification=centering}
    \centering
    \begin{subfigure}{\textwidth}
    \centering
    \includegraphics[width=\textwidth]{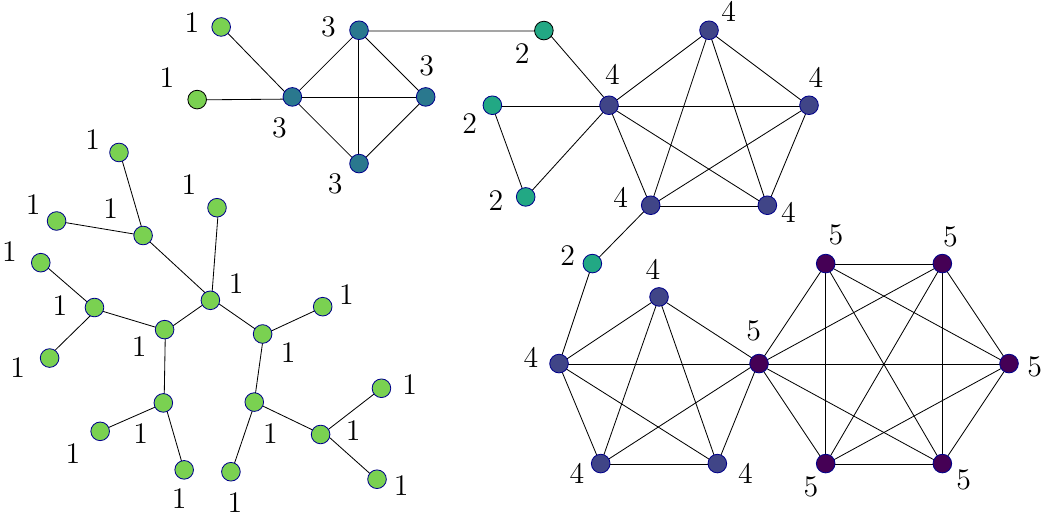}
    \subcaption{}\label{fig:example-node-link}
    \end{subfigure}
    \hspace{0.05cm}
    \begin{subfigure}{\textwidth}
    \centering
    \includegraphics[width=\textwidth]{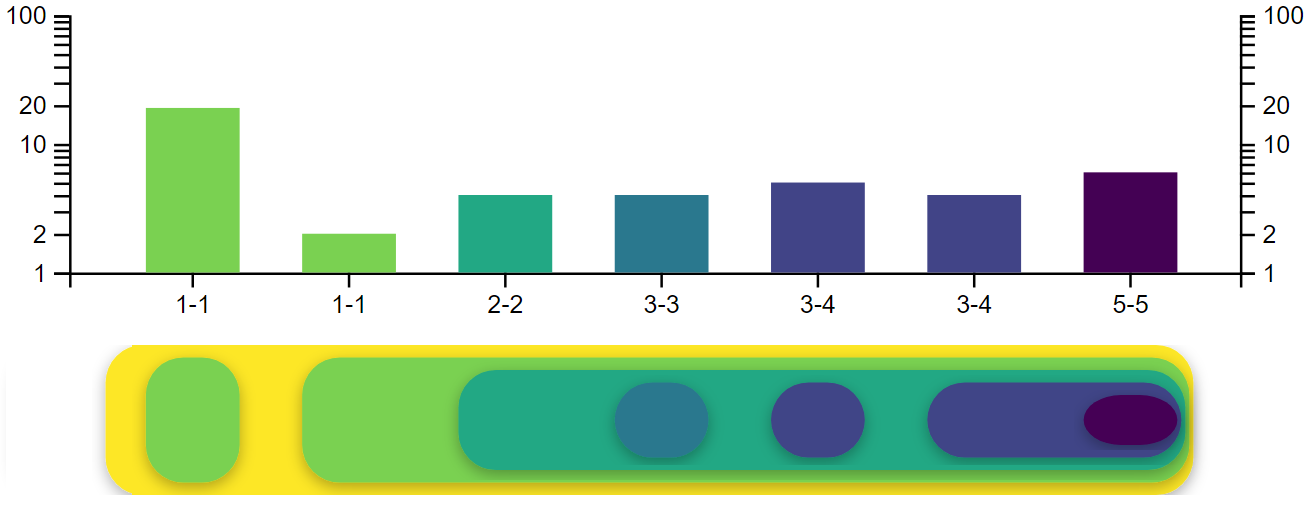}
    \caption{}\label{fig:example-1}
    \end{subfigure}
    \caption{(a) A node-link representation of a graph. Vertices are labeled with their coreness. (b) The treebar map of the graph in (a). The bar-chart is at the top and the horizontal-treemap at the bottom. }\label{fig:my_label1}
\end{figure}

More formally, a treebar map is a pair $\langle\cal H, \cal B\rangle$, where $\cal H$ is a horizontal treemap and $\cal B$ is a bar-chart, defined as follows. Let $\mu$ be a node of $T$. We have seen in \cref{sec:core-connectivity-trees} that, in general, $\mu$ is associated with a range of coreness values. This happens when it represents a $k$-core that does not change for several consecutive values of $k$. Let $\min_c(\mu)$ and by $\max_c(\mu)$ be the minimum and maximum coreness associated with $\mu$, respectively. Let $V(\mu)$ the set of vertices of any of the $k$-cores of $G$ corresponding to $\mu$. We define $n_\mu=|V(\mu)|$.  Further, suppose that $\mu$ has at least one child and let $\mu_1, \dots, \mu_h$ ($h \geq 1$) be the children of $\mu$. Observe that: (1) since we have removed from $T$ all its leaves, $T$ can contain nodes with exactly one child and hence $h$ may be equal to $1$; (2) for each $\mu_i$ we have $\min_c(\mu_i)=\max_c(\mu)+1$. We define $V^{-}(\mu)$ as the subset of vertices of $V(\mu)$ whose coreness is ``not enough'' to enter in any of $V(\mu_1), \dots, V(\mu_h)$. Formally, a vertex $v \in V(\mu)$ belongs to $V^{-}(\mu)$ if $c(v)<\min_c(\mu_1)$. We define $n_\mu^{-}=|V^{-}(\mu)|$. We have that $n_\mu^{-}=n_\mu - \sum_{i=1}^h n_{\mu_i}$.
We order the children of each node of $T$ in ascending order based on the height of their subtrees.

The horizontal tree-map $\cal H$ is recursively defined as an arrangements of rectangles, as shown in \cref{fig:example-1}. {\bf Base case.} If $T$ consists of a single node $\mu$, then $\mu$ is represented with a unit square. {\bf Inductive case.} Let $\mu$ be the root of $T$. Since we are not in the base case, $\mu$ has at least one child. Let $\mu_1, \dots, \mu_k$ ($k \geq 1$) be the children of $\mu$. Then $\mu$ is represented with a rectangle consisting of the horizontal alignment of: a unit square representing $V^{-}(\mu)$ plus the rectangles representing $\mu_1, \dots, \mu_k$, ordered according to the order of $T$. From this definition it descends that the unit squares belonging to $\cal H$ either correspond to leaves of $T$ or correspond to sets $V^{-}(\mu)$ for all non-leaf nodes $\mu$ of $T$.  In \cref{fig:example-1}, the tree-map represents a graph with two $1$-cores. One of the $1$-cores contains a $2$-core, which in turn contains three $3$-cores. The second and third are also $4$-cores. Further, the third one contains a $5$-core. The labels on the tree-map indicate the coreness range of each set of vertices, and the rectangles are colored on a scale according to their maximum coreness. Also, the rectangles are slightly resized to emphasize the inclusion relationships.

We can now define the bar-chart $\cal B$ positioned above $\cal H$. Each bar in $\cal B$ corresponds to either a leaf $\mu$ of $T$ or to a set $V^{-}(\mu)$ for some non-leaf node $\mu$ of $T$. In the first case, the height of the bar is logarithmic in $n_\mu$, while in the second case the height is logarithmic in $n_\mu^{-}$. 
In the bar-chart of \cref{fig:example-1} we observe that the first $1$-core has $19$ vertices. The second $1$-core contains $2$ vertices with coreness $2$, while all its other vertices are part of a $2$-core. This $2$-core has $4$ vertices with coreness $2$, and the remaining vertices contribute to deeper $k$-cores. The deepest is a $5$-core with $6$ vertices.
Note that in $\mathcal{B}$, no logarithms of sums are used. As a result, all the figures can be compared in a fair manner.

Observe that the width of $\cal H$ is directly proportional to the number of nodes in the core-connectivity tree $T$, while its height remains fixed at one. This aspect poses a challenge in creating one-page schematic representations for graphs that have a core-connectivity tree $T$ with a significant number of nodes, such as exceeding 25-30 nodes. As the number of nodes increases beyond these thresholds, the spatial constraints of fitting such a wide structure within a one-page drawing become unsustainable and impractical.
Considering the findings that will be presented in \cref{tab:times}, where the number of nodes in $T$ can reach the order of thousands within our range of interest, it becomes clear that the metaphor presented thus far does not adequately address Challenge CT.

In order to achieve a scalable representation, we employ a visual simplification technique similar to what is used in geographic maps. This simplification becomes evident when examining \cref{fig:example-1}. In this figure, the curves separating one set from another can be interpreted as \emph{contour lines} in a physical map, representing changes in coreness height.
To scale the representation, instead of drawing contour lines at every change of $1$ unit of coreness, we can use contour lines at larger intervals, such as every $2$ units. This is illustrated in \cref{fig:example-2}. In this simplified representation, the transitions between coreness $0$ and $1$, $2$ and $3$, and $4$ and $5$ are no longer visible. While the fine details of the first small hill with coreness $1$ from \cref{fig:example-1} are lost, the overall view and the structure of the graph are preserved.
Importantly, users are familiar with the concept of contour lines and understand that their accuracy can vary depending on the scale of the map. This allows them to interpret the simplified representation accordingly. By leveraging this analogy to geographic maps, we strike a balance between scalability and preserving the key insights of the graph's structure.
We call \emph{coreness scale} (denote by $1:t$) the inverse of the number of coreness units between consecutive contour lines of the treemap. While the coreness scale of \cref{fig:example-2} is $1 : 2$, \cref{fig:example-4} shows a treebar map of the graph of \cref{fig:example-node-link} with coreness scale $1 : 4$.

\begin{figure}[tb]
    \captionsetup[subfigure]{justification=centering}
    \centering
    \begin{subfigure}{0.55\textwidth}
    \centering
    \includegraphics[page=1, width=\textwidth]{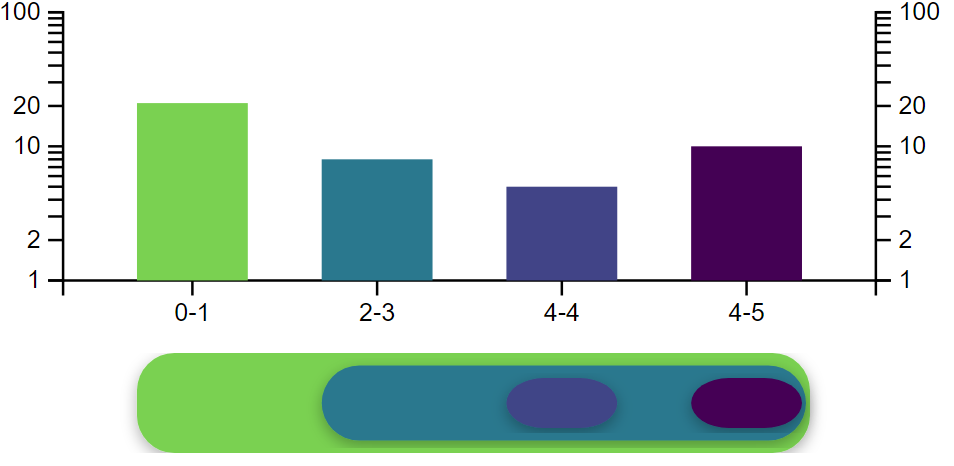}
    \subcaption{}\label{fig:example-2}
    \end{subfigure}
    \hspace{0.05cm}
    \begin{subfigure}{0.43\textwidth}
    \centering
    \includegraphics[page=2, width=\textwidth]{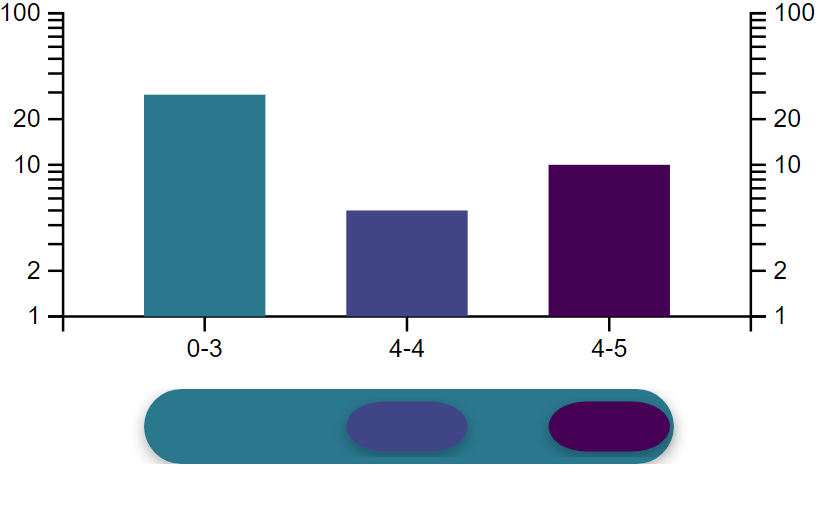}
    \caption{}\label{fig:example-4}
    \end{subfigure}
    \caption{(a) A treebar map of the graph shown in \cref{fig:example-node-link}. The contour lines span two levels of coreness (coreness scale $= 1:2$). The label 4-4 under the third bar indicates that there are no vertices of coreness $5$ in the represented set. (b) A treebar map of the same graph with coreness scale $=1:4$. }\label{fig:my_label}
\end{figure}

More formally, given a core-connectivity tree $T$ and an integer $t>1$, representing the number of coreness units that we want to collapse in the representation, we can process $T$ as follows to obtain a new tree $T'$ that can be represented using the metaphor described above.
To ensure clarity, let us assume that all nodes $\mu$ of $T$ satisfy $\min_c (\mu)=\max_c (\mu)$ (refer to \cref{fig:tree}). If this is not the case ($\min_c (\mu)<\max_c (\mu)$), just replace $\mu$ with a chain $\mu_1, \dots, \mu_k$ consisting of $k = \min_c (\mu) - \max_c (\mu)+1$ nodes, where $\min_c (\mu_i)=\max_c (\mu_i)=\min_c (\mu)+i-1$ ($i=1, \dots, k$). Next, we divide $T$ in layers. The first layer $L_1$ includes nodes $\mu$ with $0 \leq \min_c (\mu) \leq t-1$. The second layer $L_2$ includes nodes $\mu$ with $t \leq \min_c (\mu) \leq 2t-1$. In general, layer $L_i$ includes nodes $\mu$ with $(i-1) \cdot t \leq \min_c (\mu) \leq i \cdot t-1$, until $i \cdot t-1$ is greater or equal than the maximum coreness of $T$. Now, consider the subgraph $F_i$ of $T$ induced by the nodes in layer $L_i$. We observe that $F_i$ is a forest (see \cref{fig:tree-1}). For each tree $T_{ij} \in F_i$ We collapse $T_{ij}$ into a single node $\nu_j$ (see \cref{fig:tree-2}). The parent of $\nu_j$ is the parent of the root of $T_{ij}$ in $T$, if it exists. The children of $\nu_j$ are the children in $T$ of all the leaves of $T_{ij}$. Additionally, we have that $\min_c(\nu_j)=(i-1) \cdot t = \min_{\nu \in T_{ij}} \min_c(\nu)$ and $\max_c (\nu_j)= \max_{\nu \in T_{ij}} (\max_c(\nu))$. Once $T'$ has been computed, we collapse chains of nodes with only one child into a single node, updating its minimum and maximum coreness accordingly. Finally, we represent $T'$ with exactly the same metaphor described above.

\begin{figure}[t]
    \captionsetup[subfigure]{justification=centering}
    \centering
    \hfil
    \begin{subfigure}{0.25\textwidth}
    \centering
    \includegraphics[page=1, width=\textwidth]{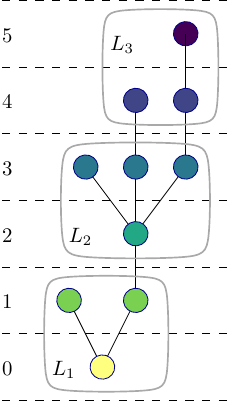}
    \subcaption{}\label{fig:tree-1}
    \end{subfigure}
    \hfil
    \begin{subfigure}{0.25\textwidth}
    \centering
    \includegraphics[page=2, width=\textwidth]{tree.pdf}
    \caption{}\label{fig:tree-2}
    \end{subfigure}
    \hfil
    \caption{(a) Tree $T$ for the graph depicted in \cref{fig:example-node-link}. (b) The same tree after aggregation to generate the treebar map with coreness scale $= 1:2$. }\label{fig:tree}
\end{figure}

\section{Leveraging Treebar Maps for the Visual Analysis of Large Graphs}\label{sec:experiments}

In this section, we demonstrate the efficiency and the effectiveness of treebar maps in providing a concise representation of the inner structure of graphs. We conducted our experiments using a library of graphs provided by Konect~\cite{konect}. Specifically, we created a \emph{sample set} comprising graphs with the following characteristics:
\begin{inparaenum}[(1)]
\item Number of edges ranging from ten million to two billion (we halted at the \texttt{MAXINT} value of 2,147,483,647 for the C language, but it is possible to surpass this limit through certain hacks).
\item Absence of multiple edges.
\item Non-bipartite nature.
\end{inparaenum}
Since the sample set contained a substantial number of graphs (38) derived from Wikipedia hyperlink networks, we chose to include only eight (roughly $20 \%$) of them in our analysis. 

Before computing a treebar map, all graphs underwent a preprocessing step in which self-loops were removed. Additionally, some graphs in the sample set were directed graphs, but we treated them as undirected graphs. In the preprocessing phase, we replaced bidirectional edges with a single undirected edge.

All computations were performed on a server equipped with Intel\textregistered{} Xeon\texttrademark{} Gold 6126 CPU and 768 GB of RAM.

\cref{tab:times} presents the following information for each graph $G$ in the sample set: the number $n$ of vertices, the number $m$ of edges, the number $m'$ of edges after preprocessing, the maximum vertex degree, the computation time for computing the core-connectivity tree, the maximum coreness $c(G)$, and the number of non-leaf nodes in the core-connectivity tree. Regarding the computation time, it should be noted that once the core-connectivity tree is obtained, the time required for generating the treebar map representation is negligible. Further, the system that computes the treebar map automatically proposes to the user a treebar map with a coreness scale $1:t$ such that the bars of the bar-chart are close to $30$, so that the treemap fits a page. This is done by performing a binary search on the value of $t$. The user can anyway change $t$ to any possible value.

\begin{table}[H]
\centering
\caption{Computation times and characteristics of the core-connectivity trees of the graphs in the sample set.}\label{tab:times}
\begin{adjustbox}{width=\linewidth}
\begin{tabular}{|l|r|r|r|r|r|r|r|}
\hline
\begin{minipage}[t][0.8cm][c]{2cm}
Graph $G$ 
\end{minipage}
& 
\begin{minipage}[t][0.8cm][c]{1.5cm}
\centering $n$ 
\end{minipage}
& 
\begin{minipage}[t][0.8cm][c]{1.8cm}
\centering $m$ 
\end{minipage}
& 
\begin{minipage}[t]{1.8cm}$m'$\\ after preprocessing  
\end{minipage}
& 
\begin{minipage}[t]{1.2cm}
Max\\ vertex\\ degree 
\end{minipage}
& 
\begin{minipage}[t]{0.8cm}
Time\\ (s) 
\end{minipage}
& 
\begin{minipage}[t]{0.8cm}$c(G)$
\end{minipage}& 
\begin{minipage}[t]{1cm}
\# non-leaf nodes     
\end{minipage}
\\ \hline
dimacs9-W & 6,262,104 & 15,119,284 & 7,559,642 & 9 & 9.4 & 3 & 117 \\ \hline
as-skitter & 1,696,415 & 11,095,298 & 11,095,298 & 35,455 & 10.9 & 111 & 902 \\ \hline
higgs-twitter-social & 456,626 & 14,855,842 & 12,508,413 & 51,386 & 11.2 & 125 & 282 \\ \hline
dimacs10-in-2004 & 1,382,867 & 13,591,473 & 13,591,473 & 21,869 & 11 & 488 & 4,647 \\ \hline
petster-catdog-friend & 623,754 & 13,991,746 & 13,991,746 & 80,634 & 13.4 & 419 & 823 \\ \hline
zhishi-hudong-internallink & 1,984,484 & 14,869,484 & 14,428,382 & 61,440 & 17.3 & 266 & 5,142 \\ \hline
flickr-links & 1,715,255 & 15,551,250 & 15,551,249 & 27,224 & 15.2 & 568 & 23,579 \\ \hline
petster-carnivore & 623,766 & 15,699,276 & 15,695,166 & 80,637 & 13.4 & 1159 & 8,997 \\ \hline
dimacs-10-eu-2005 & 862,664 & 16,138,468 & 16,138,468 & 68,963 & 13 & 388 & 253 \\ \hline
patentcite & 3,774,768 & 16,518,947 & 16,518,947 & 793 & 22.9 & 64 & 4,142 \\ \hline
dimacs9-CTR & 14,081,816 & 33,866,826 & 16,933,413 & 8 & 21.3 & 3 & 325 \\ \hline
libimseti & 220,970 & 17,359,346 & 17,233,144 & 33,389 & 15.4 & 273 & 273 \\ \hline
zhishi-hudong-relatedpages & 2,452,715 & 18,854,882 & 18,690,759 & 204,277 & 20.3 & 16 & 12,553 \\ \hline
wikipedia\_link\_ro & 903,416 & 32,763,547 & 21,875,288 & 139,792 & 18.1 & 887 & 3,004 \\ \hline
soc-pokec-relationships & 1,632,803 & 30,622,564 & 22,301,964 & 14,854 & 26.5 & 47 & 53 \\ \hline
flickr-growth & 2,302,925 & 33,140,017 & 22,838,276 & 27,937 & 25.2 & 600 & 34,866 \\ \hline
dimacs9-USA & 23,947,347 & 57,708,624 & 28,854,312 & 9 & 40.4 & 3 & 538 \\ \hline
wikipedia-growth & 1,870,709 & 39,953,145 & 36,532,531 & 226,073 & 43.8 & 206 & 292 \\ \hline
soc-LiveJournal1 & 4,846,609 & 68,475,391 & 42,851,237 & 20,333 & 53.9 & 372 & 2,956 \\ \hline
livejournal-links & 5,204,176 & 49,174,464 & 48,709,621 & 15,016 & 60.1 & 374 & 6,871 \\ \hline
wikipedia\_link\_ja & 1,767,268 & 83,202,622 & 65,495,572 & 274,537 & 72 & 887 & 506 \\ \hline
wikipedia\_link\_de & 3,603,726 & 96,865,851 & 77,546,982 & 434,234 & 89.1 & 837 & 1,409 \\ \hline
wikipedia\_link\_it & 2,148,791 & 104,719,994 & 77,875,131 & 286,585 & 75.7 & 899 & 520 \\ \hline
wikipedia\_link\_sv & 6,100,692 & 106,749,786 & 99,864,874 & 2,732,817 & 95.2 & 357 & 310 \\ \hline
wikipedia\_link\_sr & 3,175,009 & 139,586,199 & 103,310,837 & 369,566 & 80 & 1,869 & 718 \\ \hline
orkut-links & 3,072,441 & 117,184,899 & 117,184,899 & 33,313 & 138.3 & 253 & 253 \\ \hline
dimacs10-uk-2002 & 18,483,186 & 261,787,258 & 261,787,258 & 194,955 & 235.1 & 943 & 39,551 \\ \hline
wikipedia\_link\_en & 13,593,032 & 437,217,424 & 334,591,525 & 1,052,326 & 381.2 & 1,114 & 1,215 \\ \hline
twitter & 41,652,230 & 1,468,365,182 & 1,202,513,046 & 2,997,487 & 1828.1 & 2,488 & 2,187 \\ \hline
twitter-mpi & 52,579,682 & 1,963,263,821 & 1,614,106,187 & 3,503,677 & 2412.8 & 2,647 & 32,202 \\ \hline
friendster & 68,349,466 & 2,586,147,869 & 1,811,849,342 & 5214 & 3147.9 & 304 & 329,745 \\ 
\hline
\end{tabular}
\end{adjustbox}
\end{table}

We have computed the treebar maps for all the graphs listed in \cref{tab:times}. However, due to space limitations, some of them are included in the Appendix. Here, we discuss the features of a few examples. Let's consider \cref{fig:wikipedia_link_ja} and \cref{fig:wikipedia_link_de} as our first examples. Both graphs represent Wikipedia link networks. The treebar maps reveal that both graphs exhibit a ``telescopic'' structure, where $k$-cores with increasing values of $k$ are nested inside each other. Additionally, both graphs contain a component with a coreness of approximately 750-900, consisting of about one thousand vertices. However, the graph in Figure \ref{fig:wikipedia_link_de} has a more complex shape, featuring a large gap of coreness of more than 100 units where only four vertices are present. This implies that the vertices with coreness greater than 600 compose a very dense set upon a connected component with much lower density. A quite uniform telescopic structure is shown by the graph represented in \cref{fig:libimseti}. At this level of detail, the vertices appear to be almost uniformly distributed among the coreness levels. In all three cases, the comparability of the graphs is enhanced by the chosen order of the core-connectivity tree.

\begin{figure}[tb!]
    \centering
    \includegraphics[width=\textwidth]{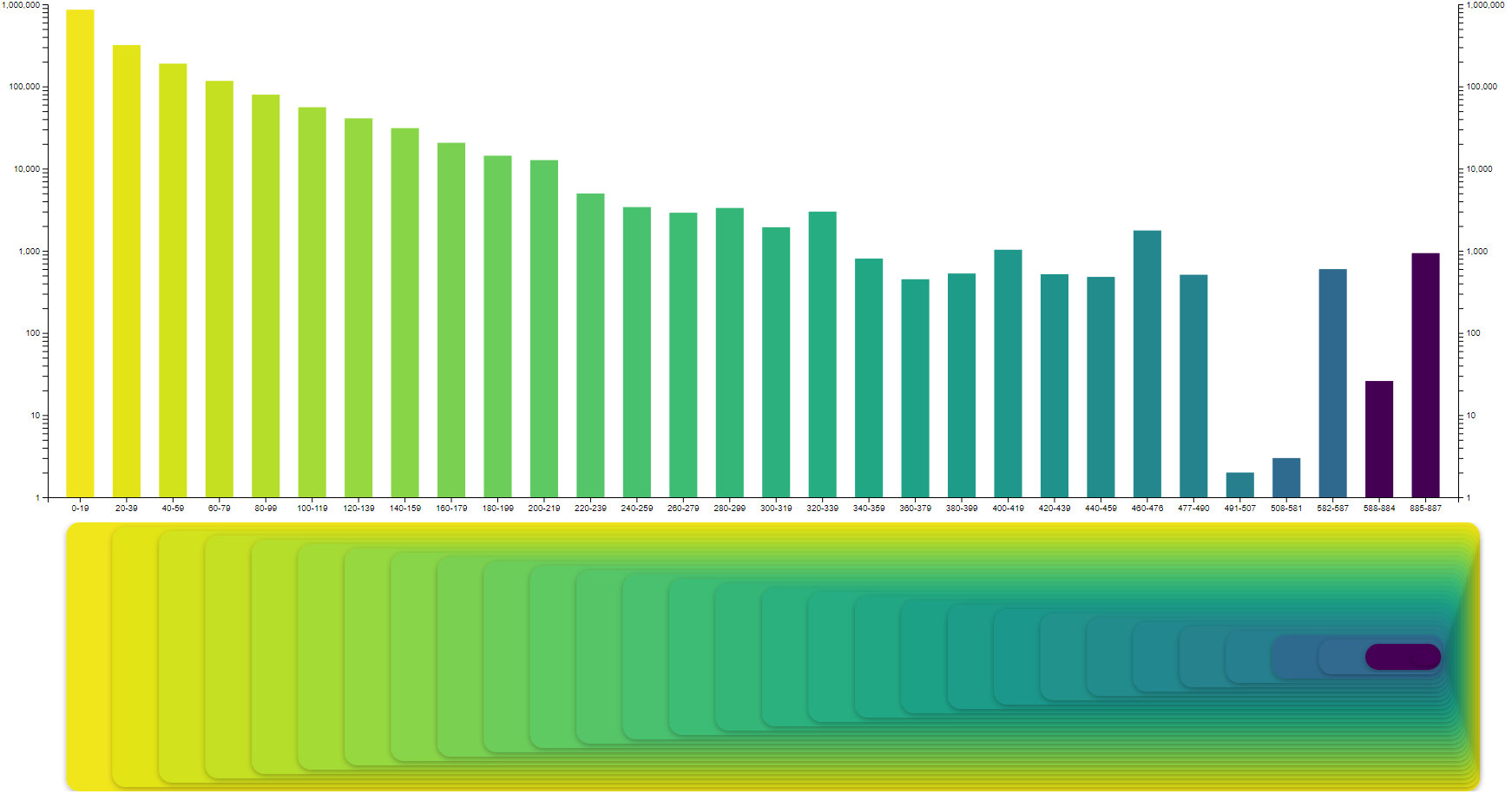}
    \caption{Network wikipedia\_link\_ja, coreness scale $=1:20$}
    \label{fig:wikipedia_link_ja}
\end{figure}

\begin{figure}[]
    \centering
    \includegraphics[width=\textwidth]{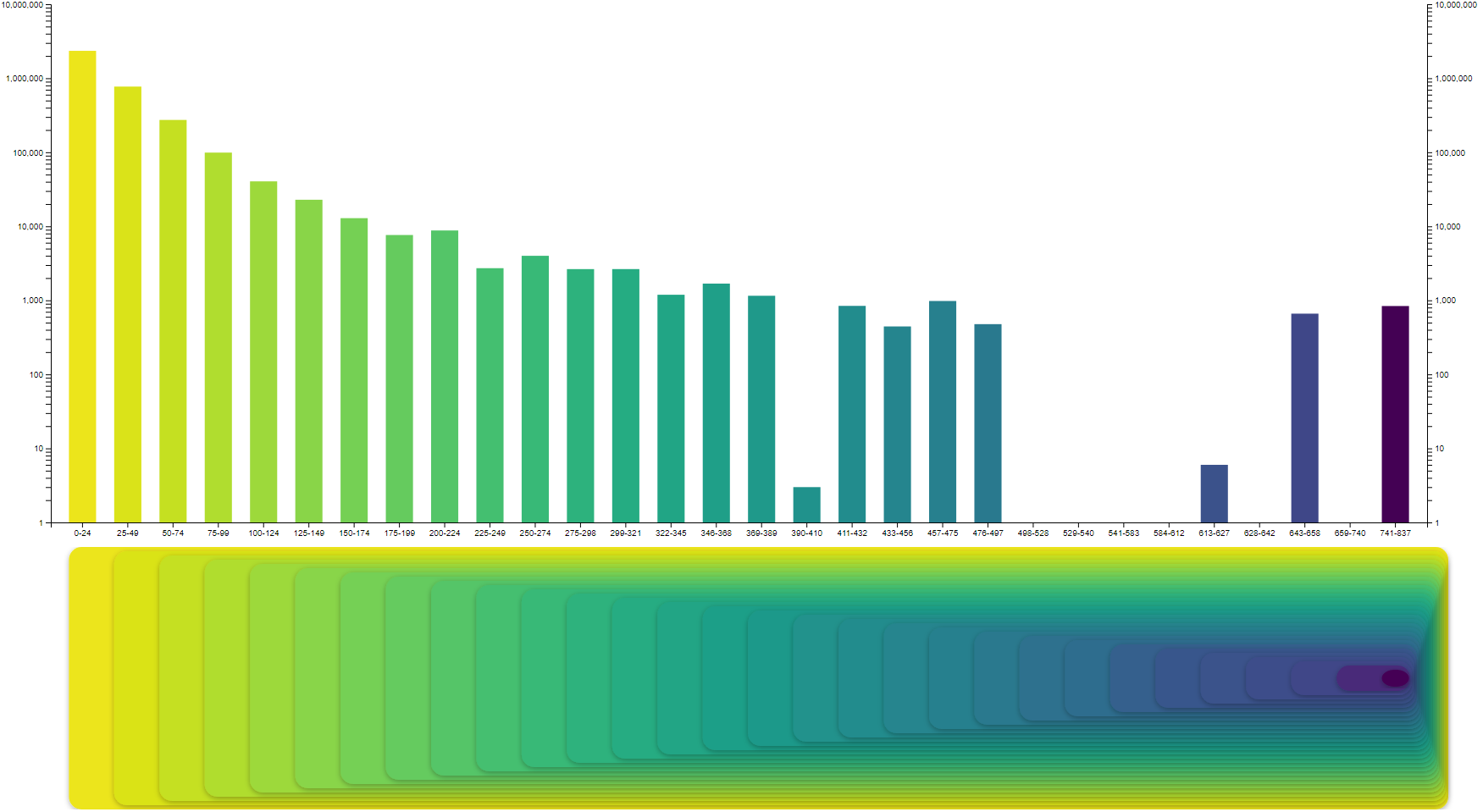}
    \caption{Network wikipedia\_link\_de, coreness scale $=1:25$}
    \label{fig:wikipedia_link_de}
\end{figure}

\begin{figure}[tb!]
    \centering
    \includegraphics[width=\textwidth]{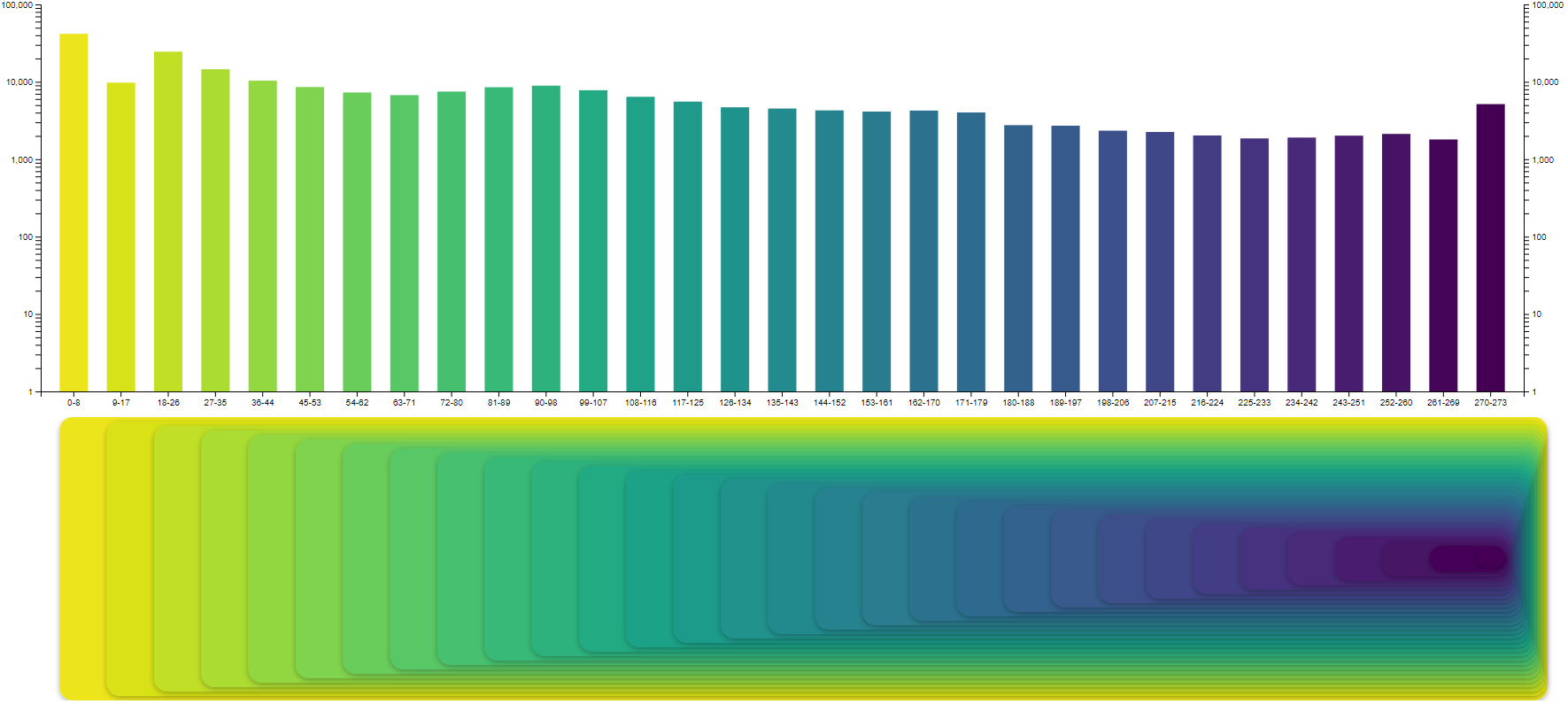}
    \caption{Network libimseti, coreness scale $=1:9$}
    \label{fig:libimseti}
\end{figure}

A contrasting situation is in \cref{fig:livejournal-links}. It shows several $k$-cores with varying values of $k$ and different numbers of vertices. Specifically, there are 13 sets representing components with coreness ranging from 81 to approx.\ 100, each consisting of around 100 vertices. Further, there are two plateaus with higher coreness. 
Within one of these plateaus, four components exhibit a somewhat telescopic~shape.

\begin{figure}[tb!]
    \centering
    \includegraphics[width=\textwidth]{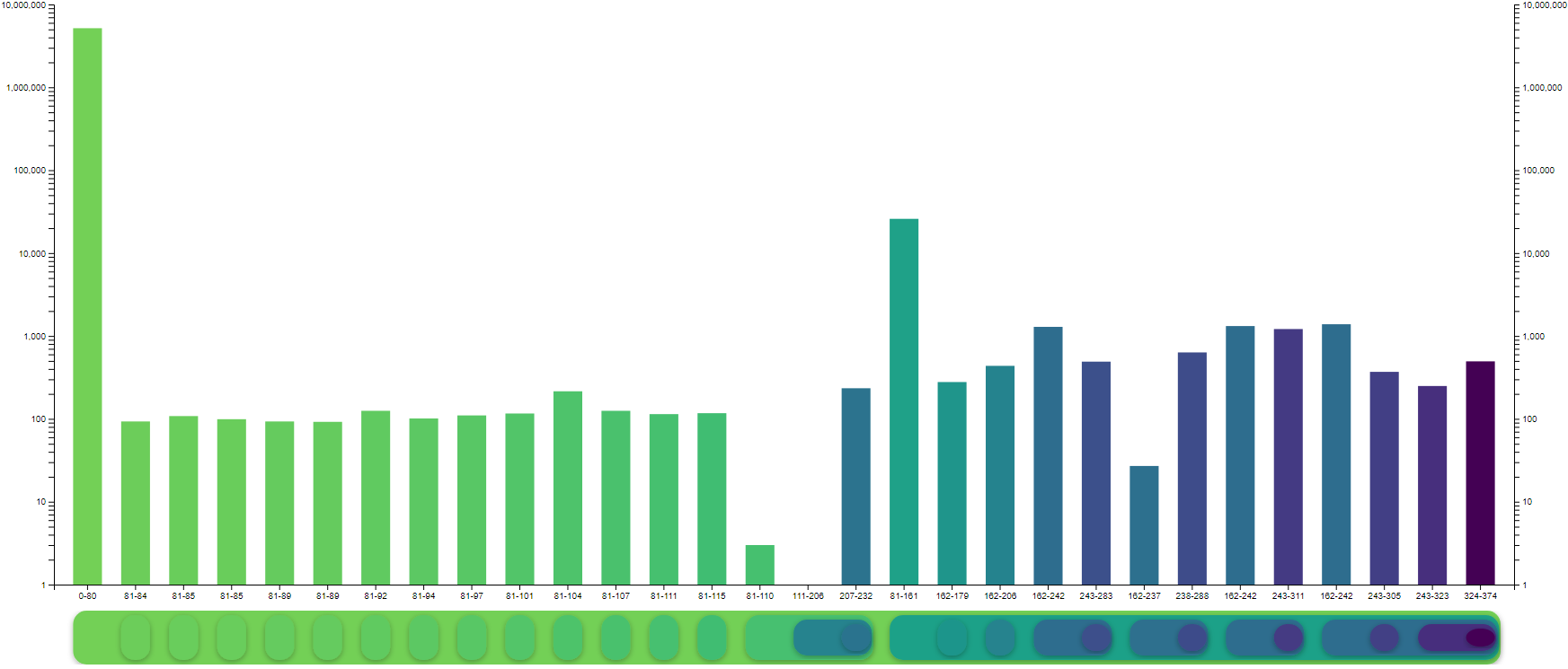}
    \caption{Network livejournal-links, coreness scale $=1:81$}
    \label{fig:livejournal-links}
\end{figure}

A similar situation happens for the graph in \cref{fig:zhishi-hudong-relatedpages}. It has essentially many distinct components of 10-100 vertices each, all having a value of coreness high enough to overcome the inverse of the coreness scale. Also there are more than two million of vertices not having enough coreness to enter such components.

\begin{figure}[tb!]
    \centering
    \includegraphics[width=\textwidth]{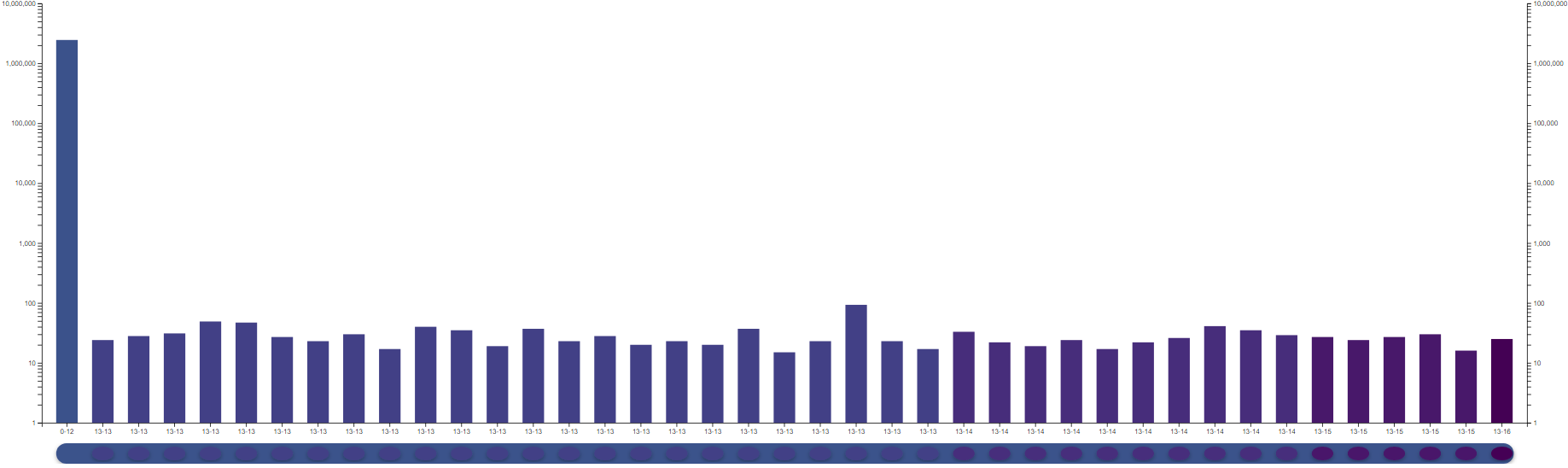}
    \caption{Network zhishi-hudong-relatedpages, coreness scale $=1:13$}
    \label{fig:zhishi-hudong-relatedpages}
\end{figure}

\begin{figure}[tb!]
    \captionsetup[subfigure]{justification=centering}
    \begin{subfigure}{0.9\textwidth}
    \includegraphics[width=\textwidth]{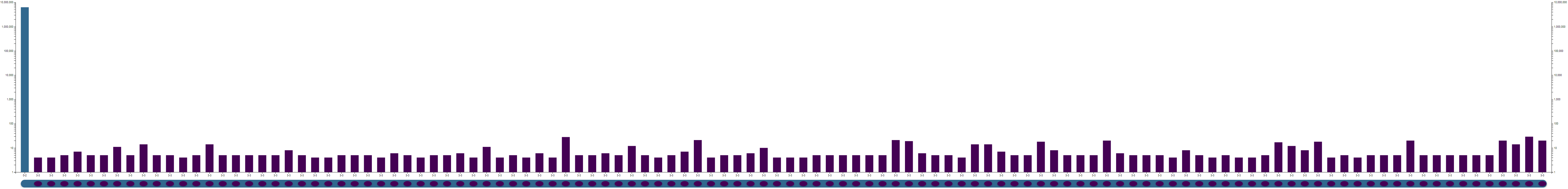}
    \subcaption{}\label{fig:dimacs9-W-scale3}
    \end{subfigure}
    \hspace{0.05cm}
    \begin{subfigure}{0.055\textwidth}
    \includegraphics[width=\textwidth]{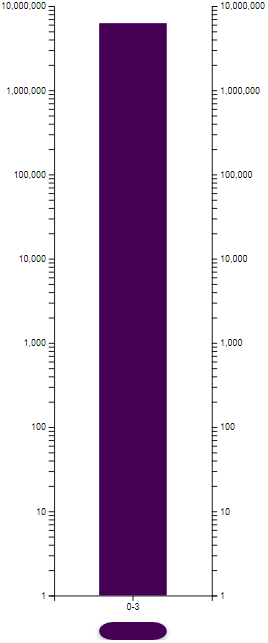}
    \caption{}\label{fig:dimacs9-W-scale4}
    \end{subfigure}
    \caption{The treebar maps for dimacs9-W with coreness scale $1:3$ (a) and $1:4$~(b).}\label{fig:dimacs9-W}
\end{figure}

\section{Conclusions and Open Problems}

We have presented a metaphor, called treebar map and based on the concept of coreness, for the schematic representation of huge graphs. Such representation can be efficiently computed exploiting a data structure called core-connectivity tree. We have shown the effectiveness of the metaphor presenting the schematic representations of the graphs contained in a widely used graph library.
Several problems deserve further investigation.
\begin{enumerate}
\item \cref{th:core-connectivity-tree} shows that the core-connectivity tree of an $n$-vertex and $m$-edges graph can be computed in $O((n+m)\alpha(n))$ time. Can this upper bound be lowered to $O(n+m)$ time?
\item The proposed metaphor reaches its limits in the example of \cref{fig:dimacs9-W}, dedicated to a graph with a peculiar structure. Despite its extensive size, with over six million vertices and seven million edges, the maximum vertex degree is merely 9, resulting in a maximum coreness of only 3. The corresponding treebar map (\cref{fig:dimacs9-W-scale3}) becomes too large to be represented with coreness scale $1:3$. Conversely, employing a coreness scale of $1:4$ (\cref{fig:dimacs9-W-scale4}) leads to the loss of significant details, rendering it less informative. We have the same behaviour for graphs dimacs9-CTR and dimacs9-USA which are omitted. It would be interesting to equip treebar maps with new features for these specific cases.
\end{enumerate}

\newpage

\bibliographystyle{splncs04}
\bibliography{bibliography}

\clearpage
\appendix

\section{Treebar Maps Produced in Our Experiments}
\begin{figure}[h!]
    \centering
    \includegraphics[width=\textwidth]{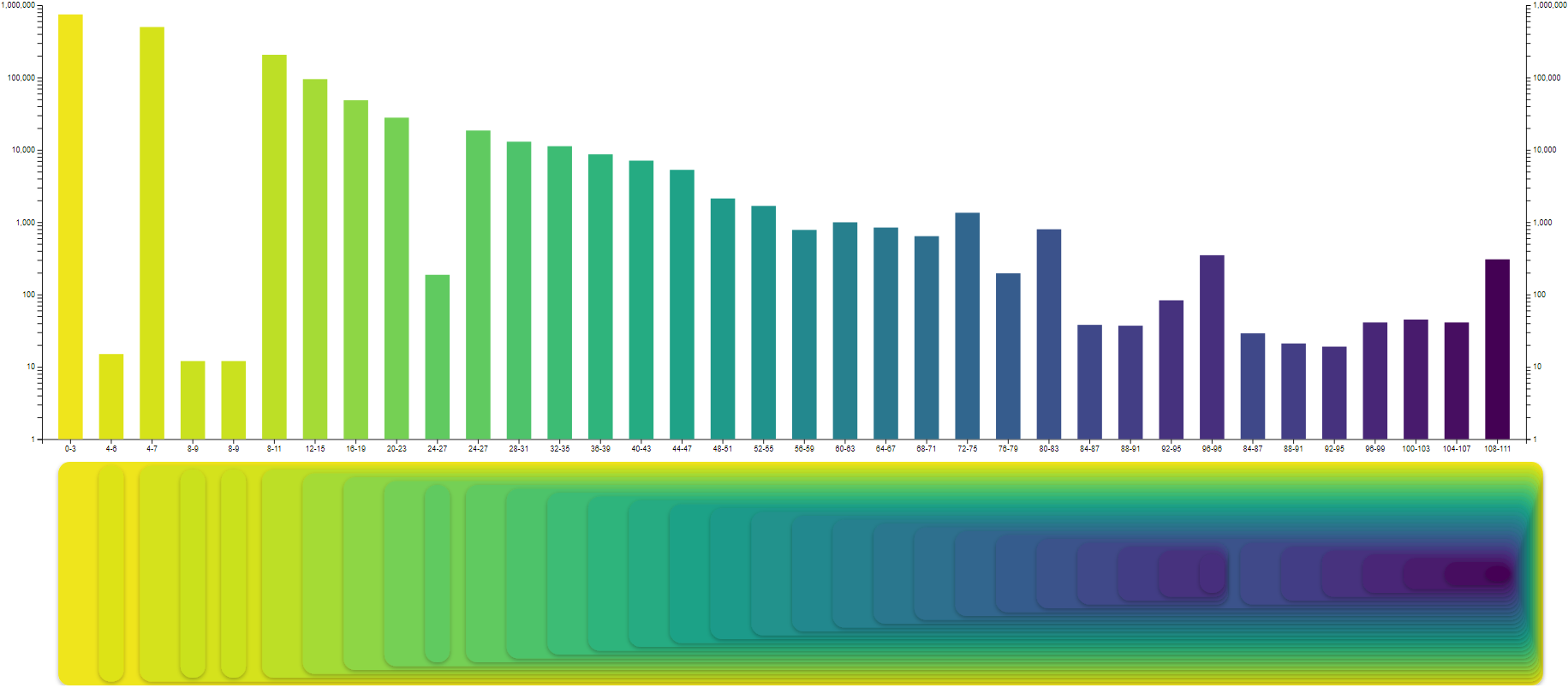}
    \caption{Network as-skitter, coreness scale $=1:4$}
    \label{fig:as-skitter}
\end{figure}
\begin{figure}[tb!]
    \centering
    \includegraphics[width=\textwidth]{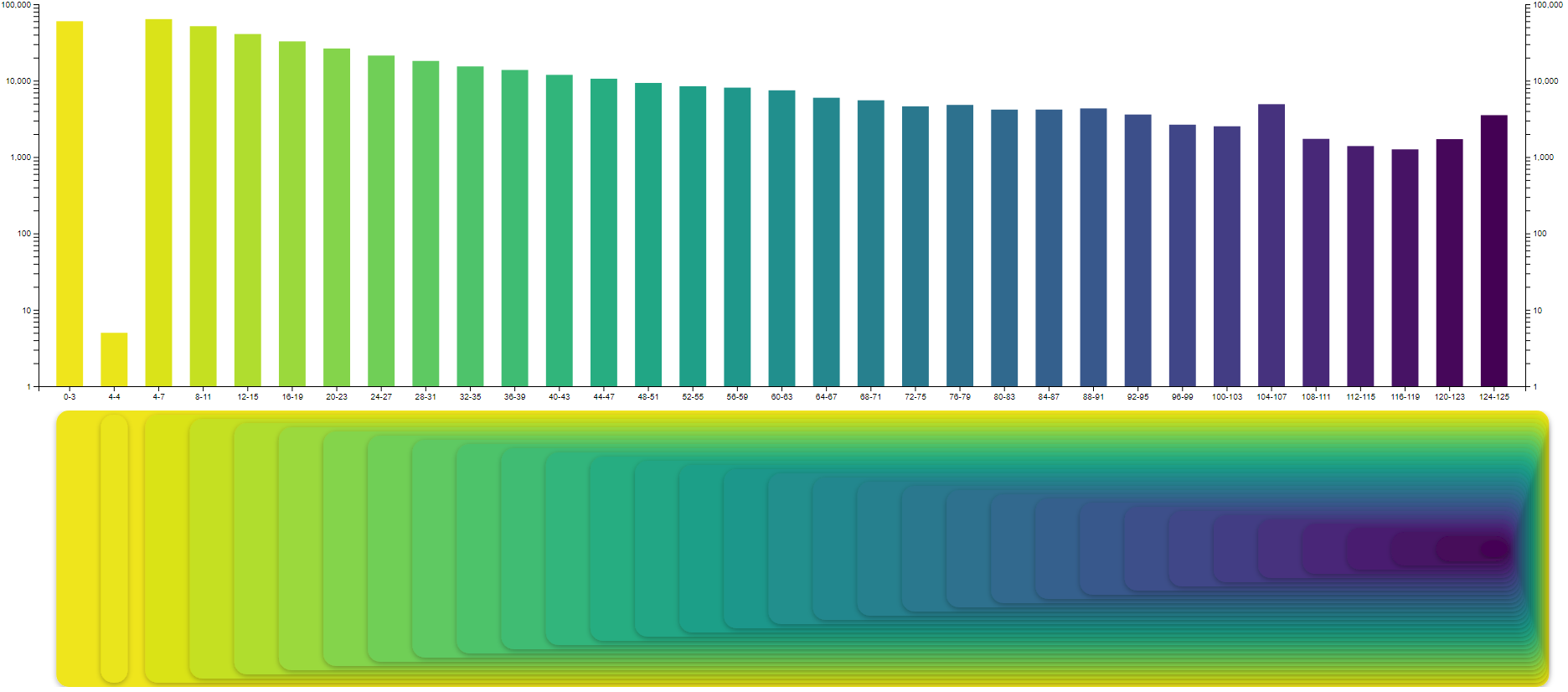}
    \caption{Network higgs-twitter-social, coreness scale $=1:4$}
    \label{fig:higgs-twitter-social}
\end{figure}
\begin{figure}[tb!]
    \centering
    \includegraphics[width=\textwidth]{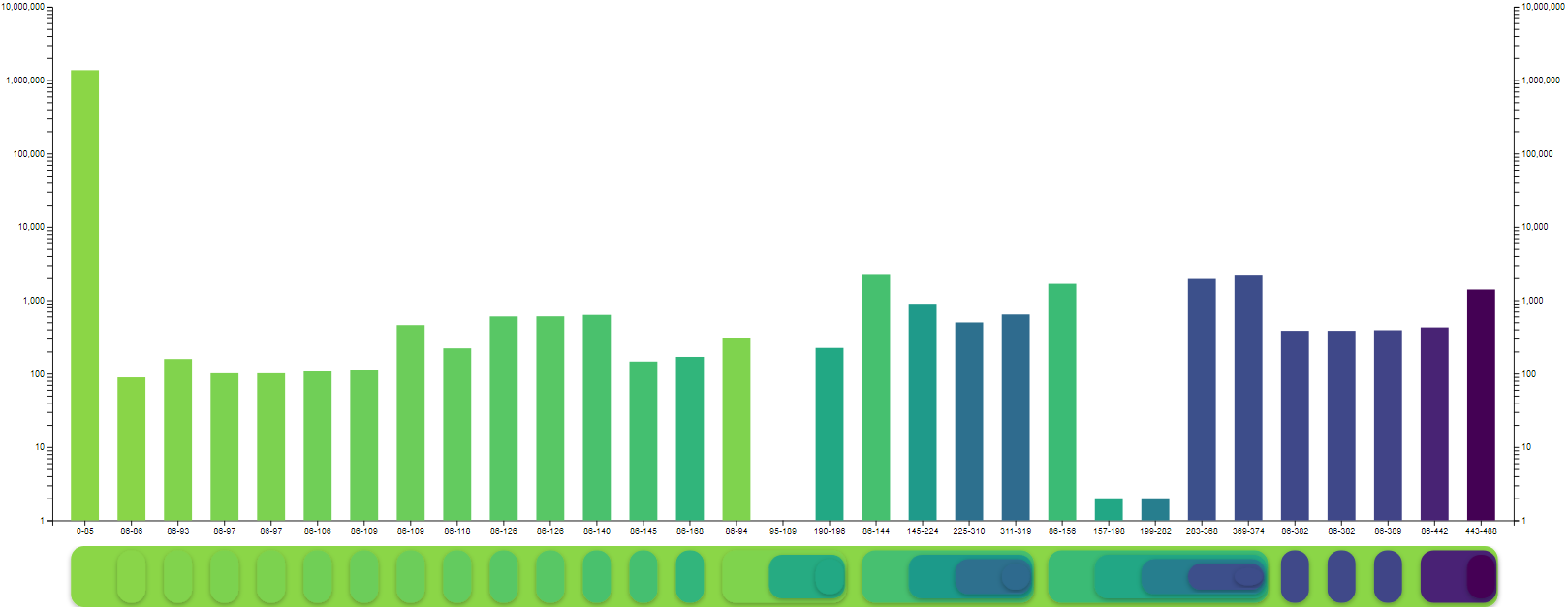}
    \caption{Network dimacs10-in-2004, coreness scale $=1:86$}
    \label{fig:dimacs10-in-2004}
\end{figure}
\begin{figure}[tb!]
    \centering
    \includegraphics[width=\textwidth]{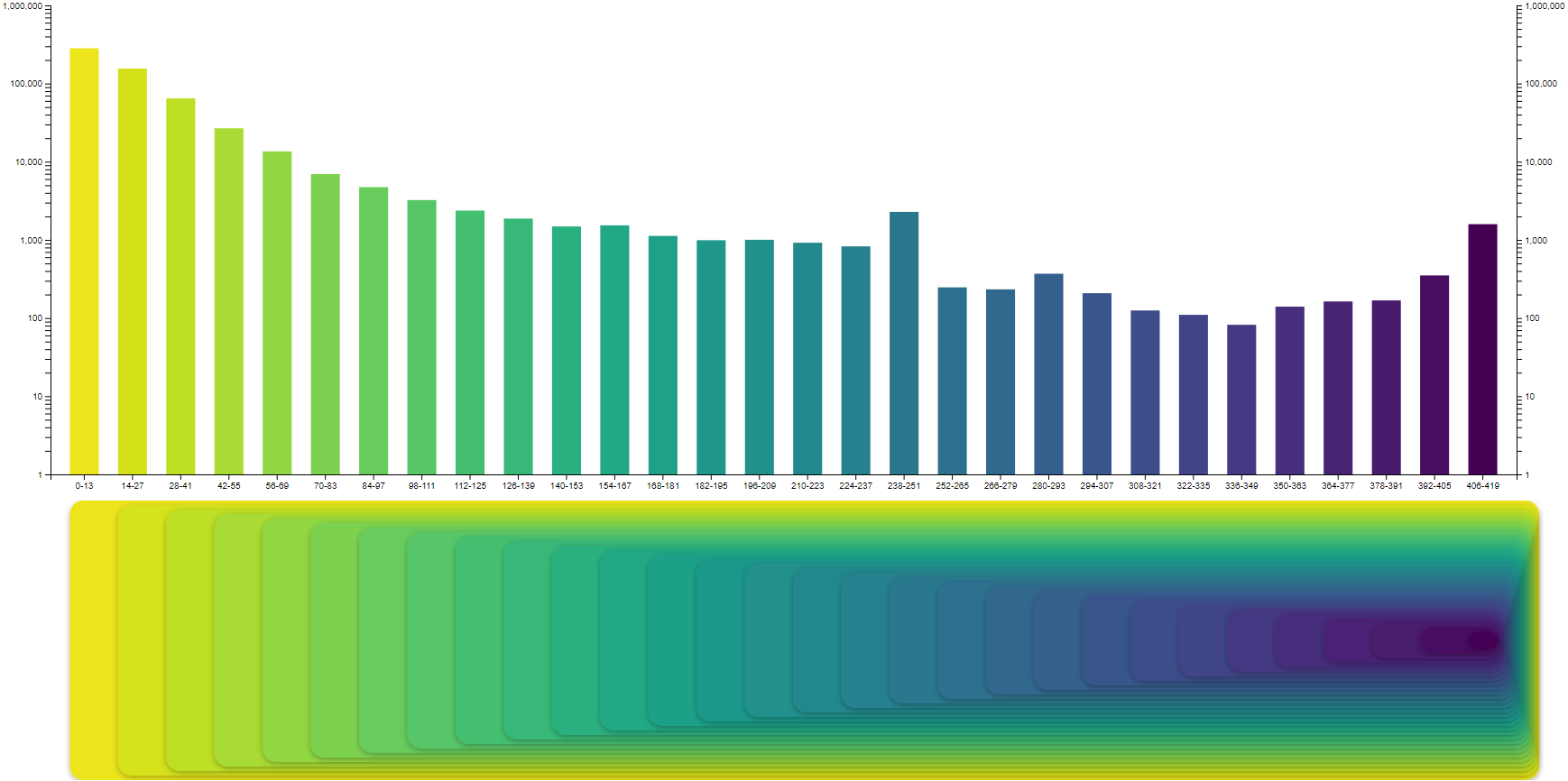}
    \caption{Network petster-catdog-friend, coreness scale $=1:14$}
    \label{fig:petster-catdog-friend}
\end{figure}
\begin{figure}[tb!]
    \centering
    \includegraphics[width=\textwidth]{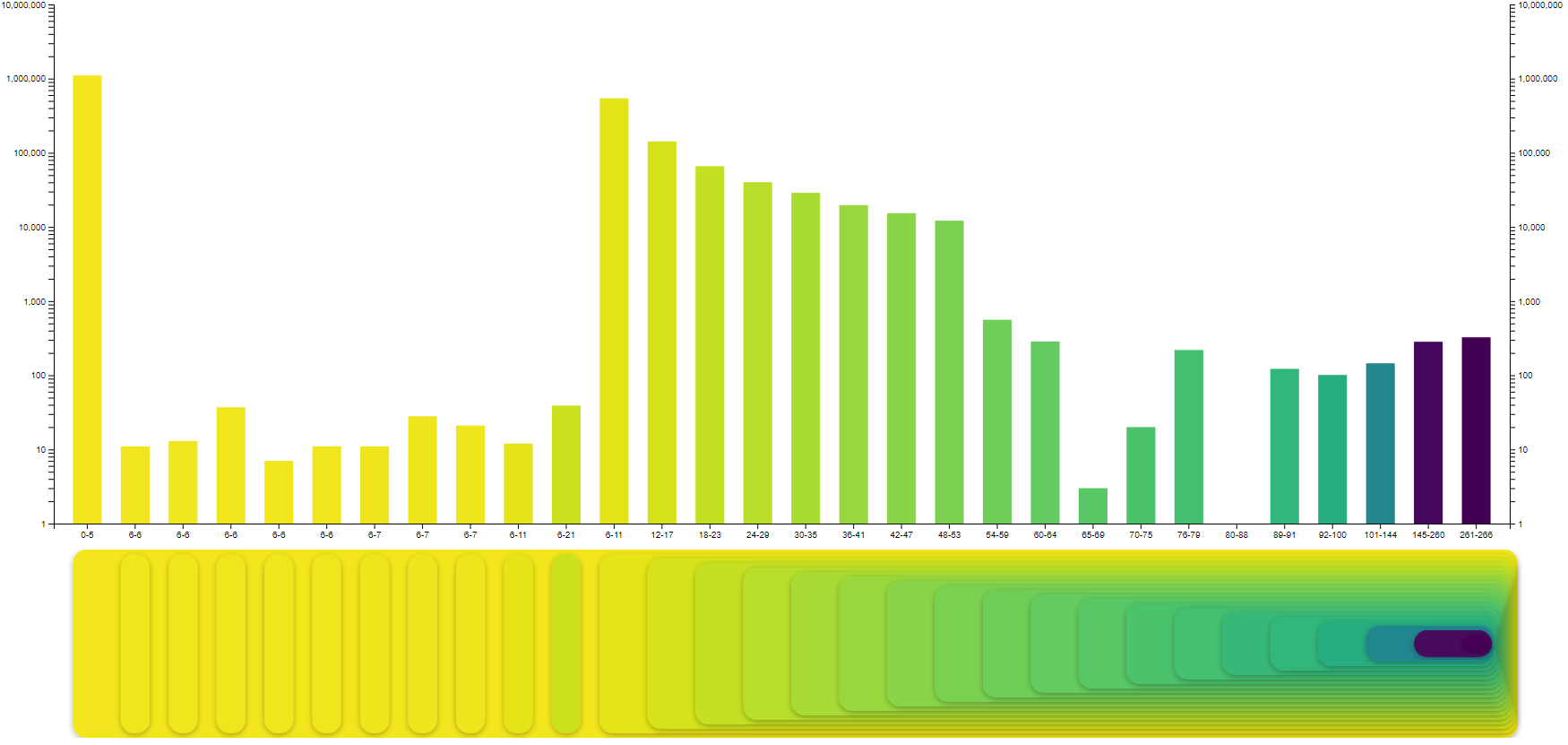}
    \caption{Network zhishi-hudong-internallink, coreness scale $=1:6$}
    \label{fig:zhishi-hudong-internallink}
\end{figure}
\begin{figure}[tb!]
    \centering
    \includegraphics[width=\textwidth]{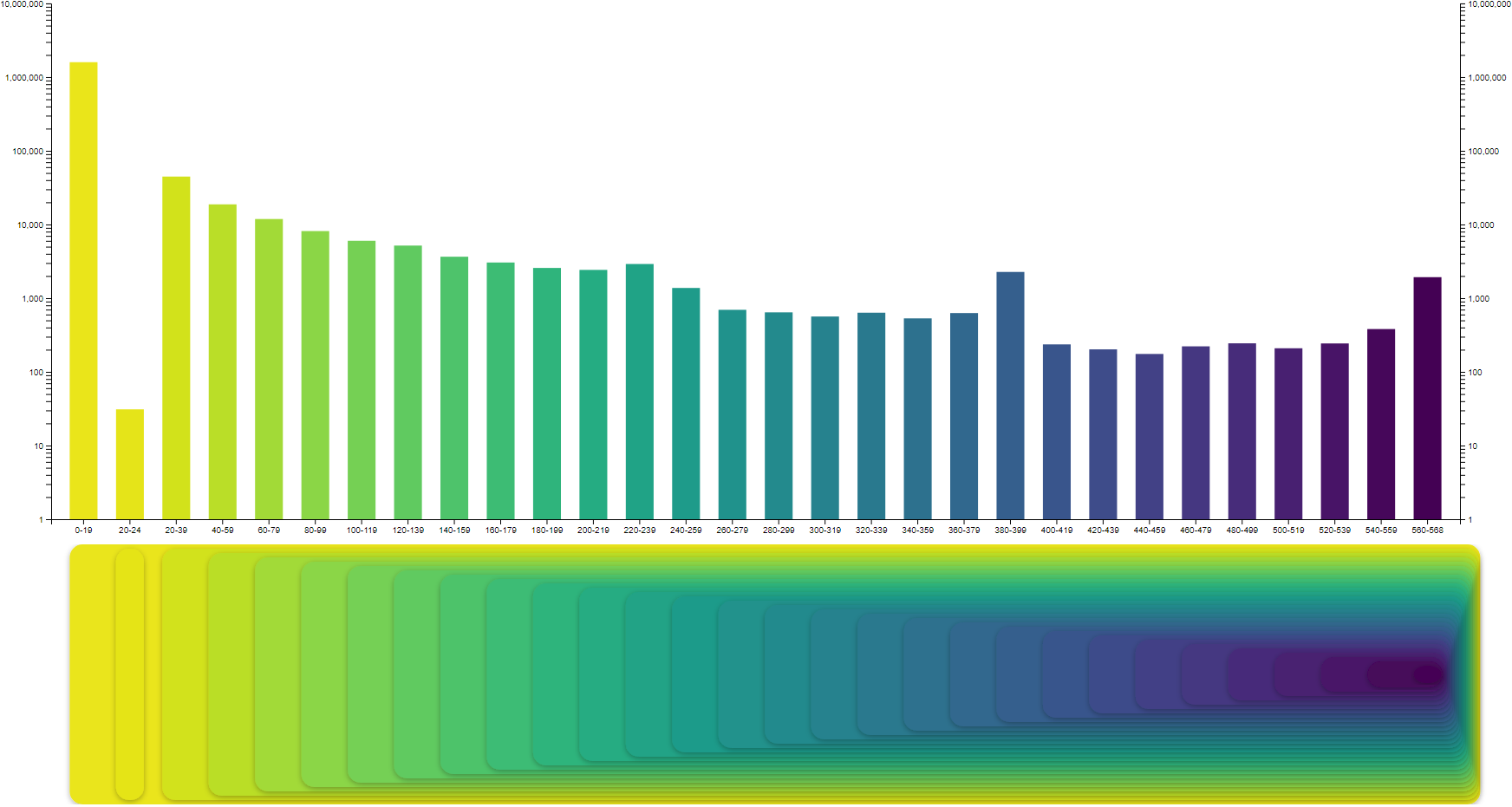}
    \caption{Network flickr-links, coreness scale $=1:20$}
    \label{fig:flickr-links}
\end{figure}
\begin{figure}[tb!]
    \centering
    \includegraphics[width=\textwidth]{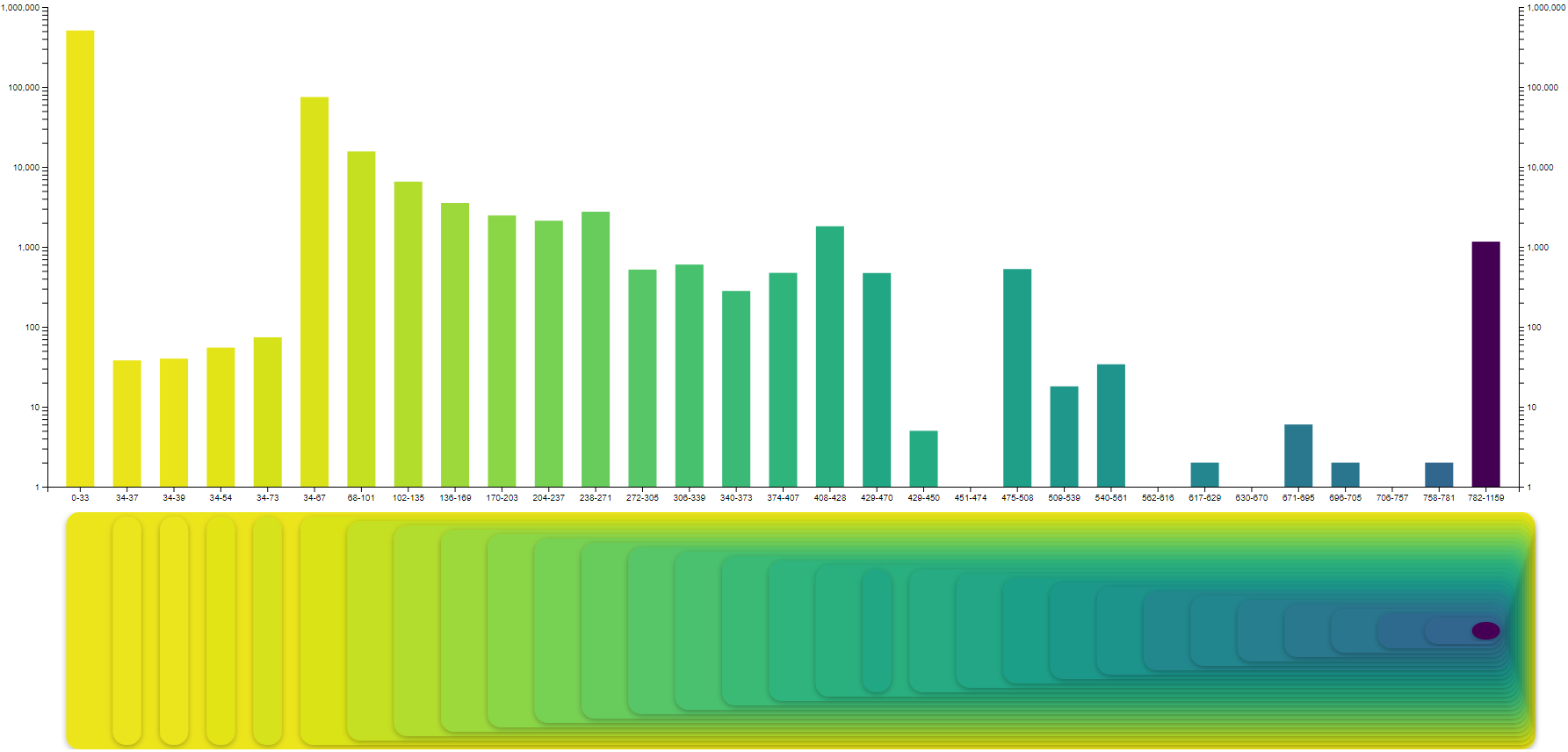}
    \caption{Network petster-carnivore, coreness scale $=1:34$}
    \label{fig:petster-carnivore}
\end{figure}
\begin{figure}[tb!]
    \centering
    \includegraphics[width=\textwidth]{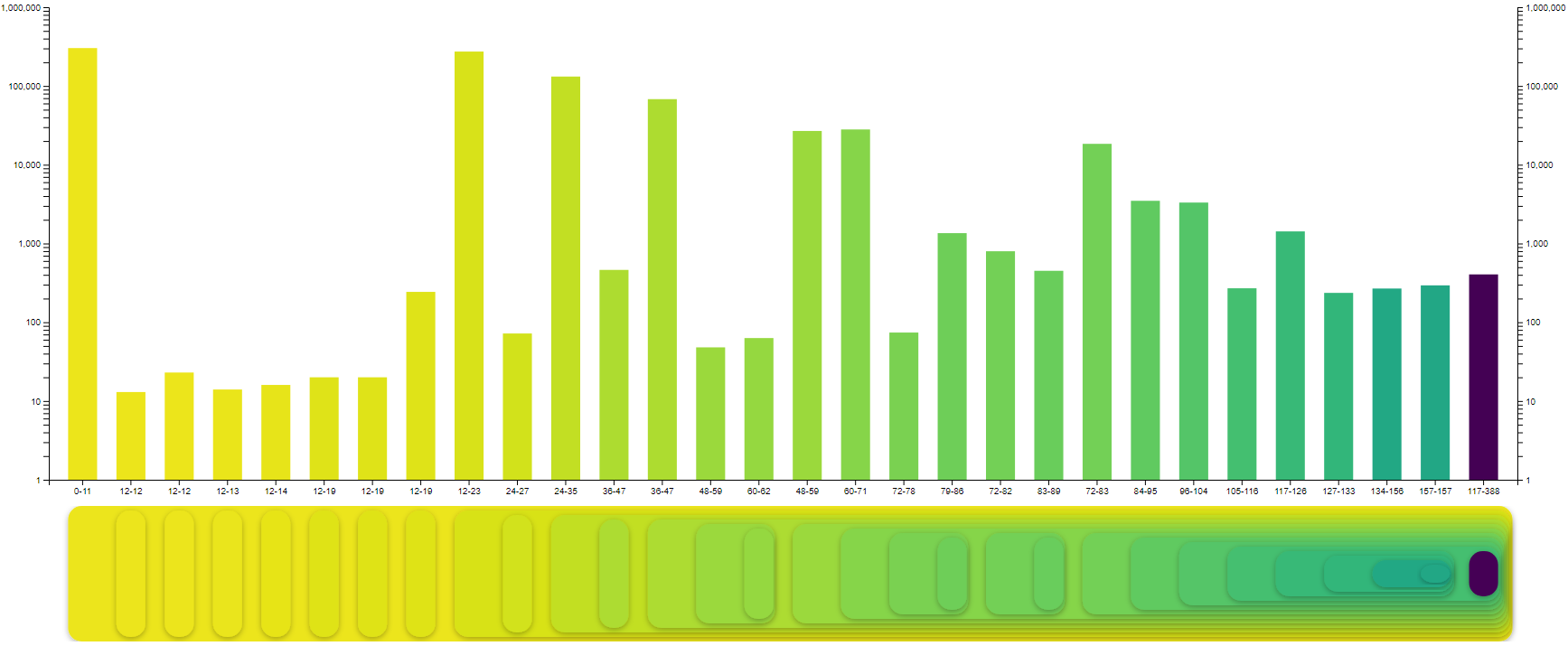}
    \caption{Network dimacs10-eu-2005, coreness scale $=1:12$}
    \label{fig:dimacs10-eu-2005}
\end{figure}
\begin{figure}[tb!]
    \centering
    \includegraphics[width=\textwidth]{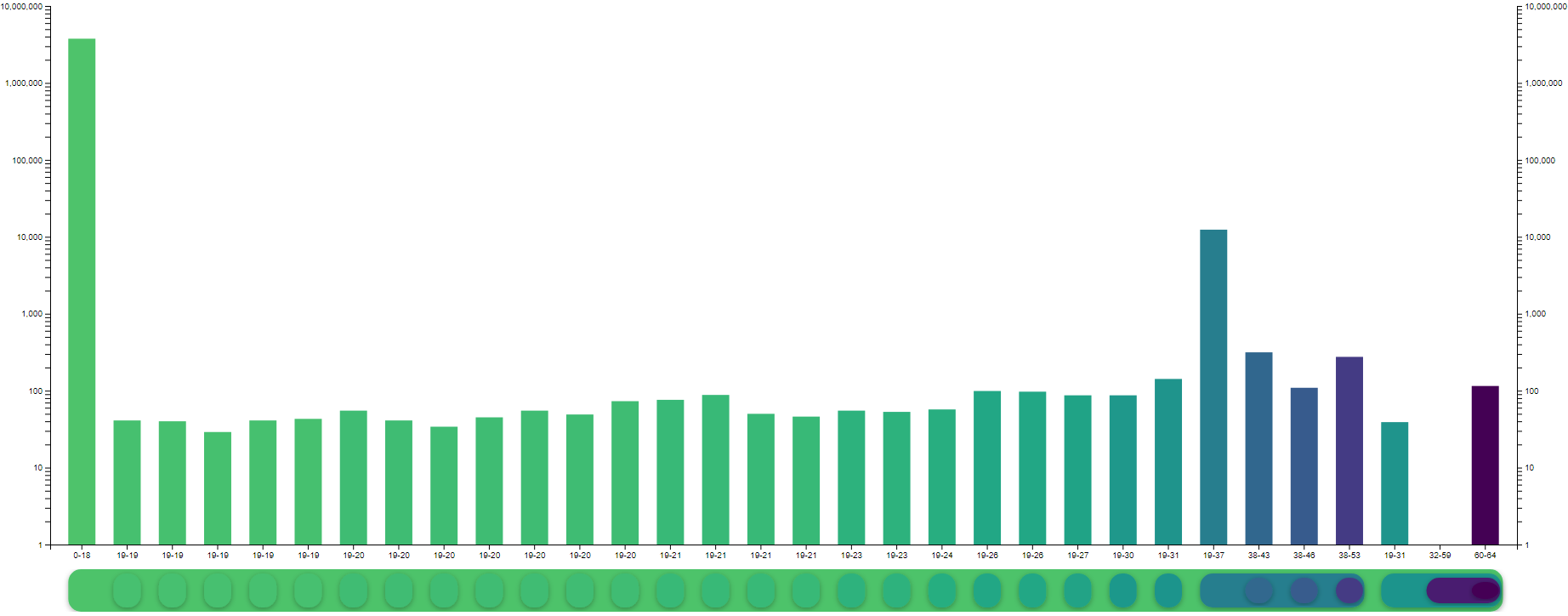}
    \caption{Network patentcite, coreness scale $=1:19$}
    \label{fig:patentcite}
\end{figure}

\begin{figure}[tb!]
    \centering
    \includegraphics[width=\textwidth]{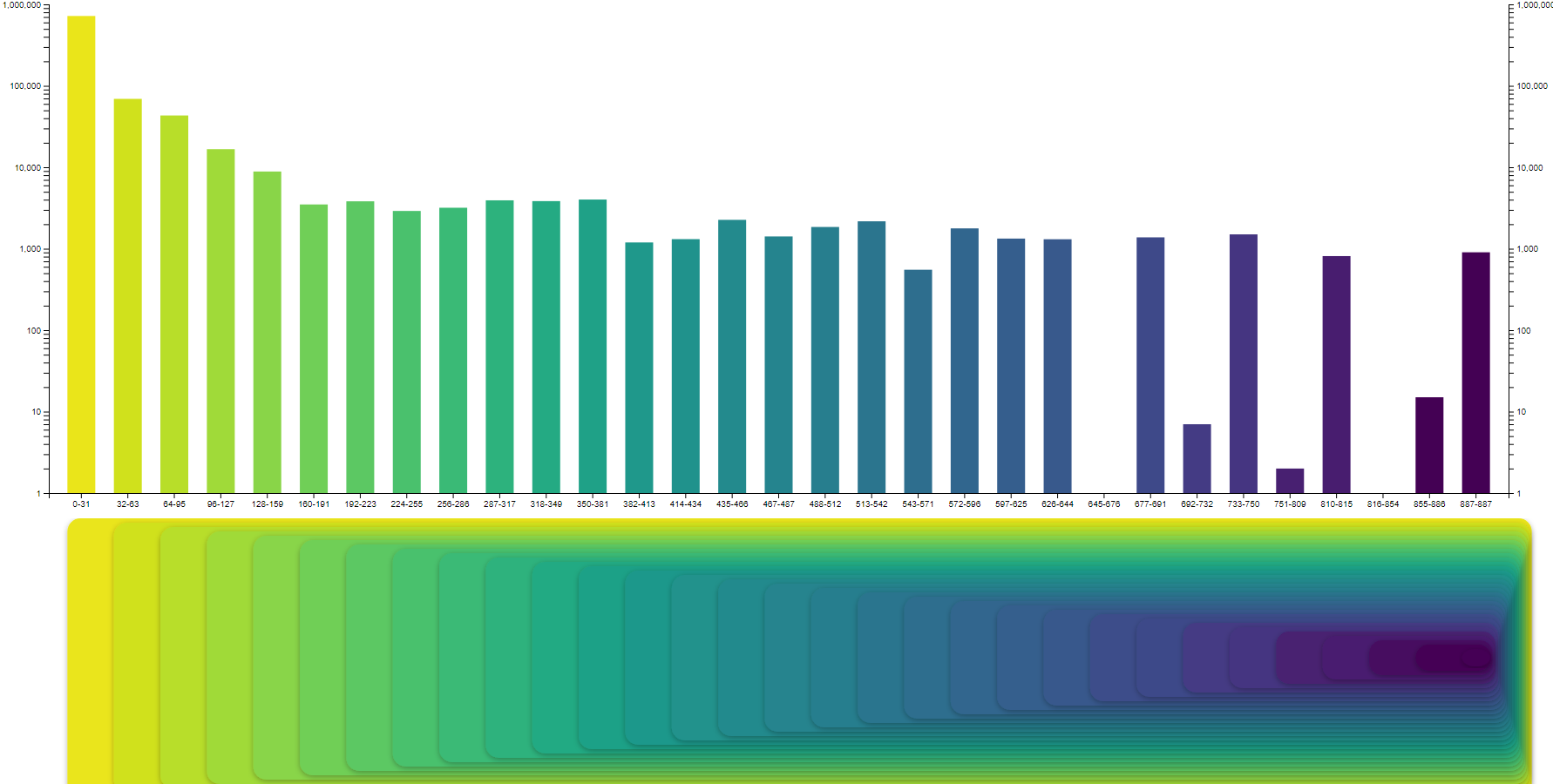}
    \caption{Network wikipedia\_link\_ro, coreness scale $=1:32$}
    \label{fig:wikipedia_link_ro}
\end{figure}
\begin{figure}[tb!]
    \centering
    \includegraphics[width=\textwidth]{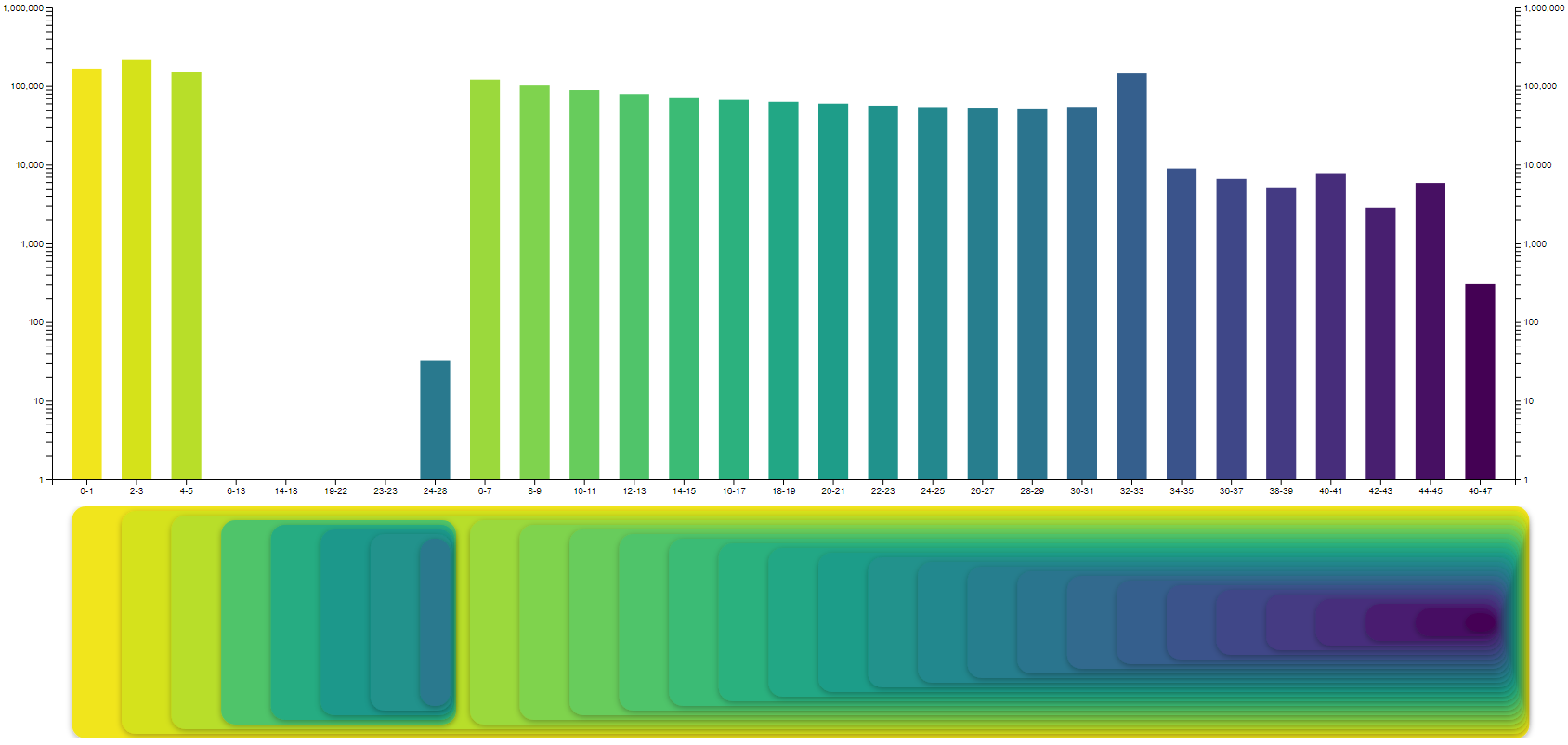}
    \caption{Network soc-pokec-relationships, coreness scale $=1:2$}
    \label{fig:soc-pokec-relationships}
\end{figure}

\begin{figure}[tb!]
    \centering
    \includegraphics[width=\textwidth]{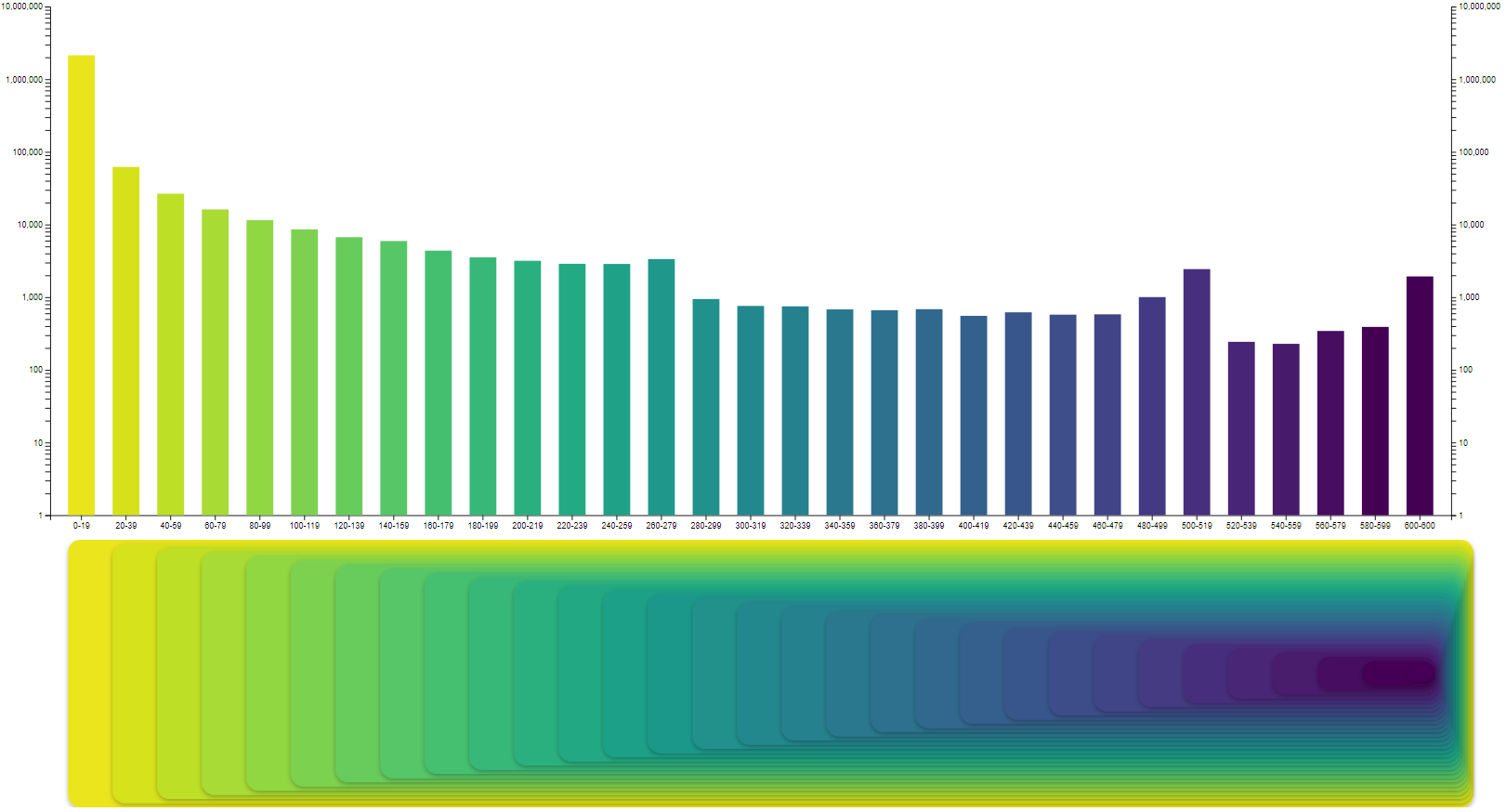}
    \caption{Network flickr-growth, coreness scale $=1:20$}
    \label{fig:flickr-growth}
\end{figure}

\begin{figure}[tb!]
    \centering
    \includegraphics[width=\textwidth]{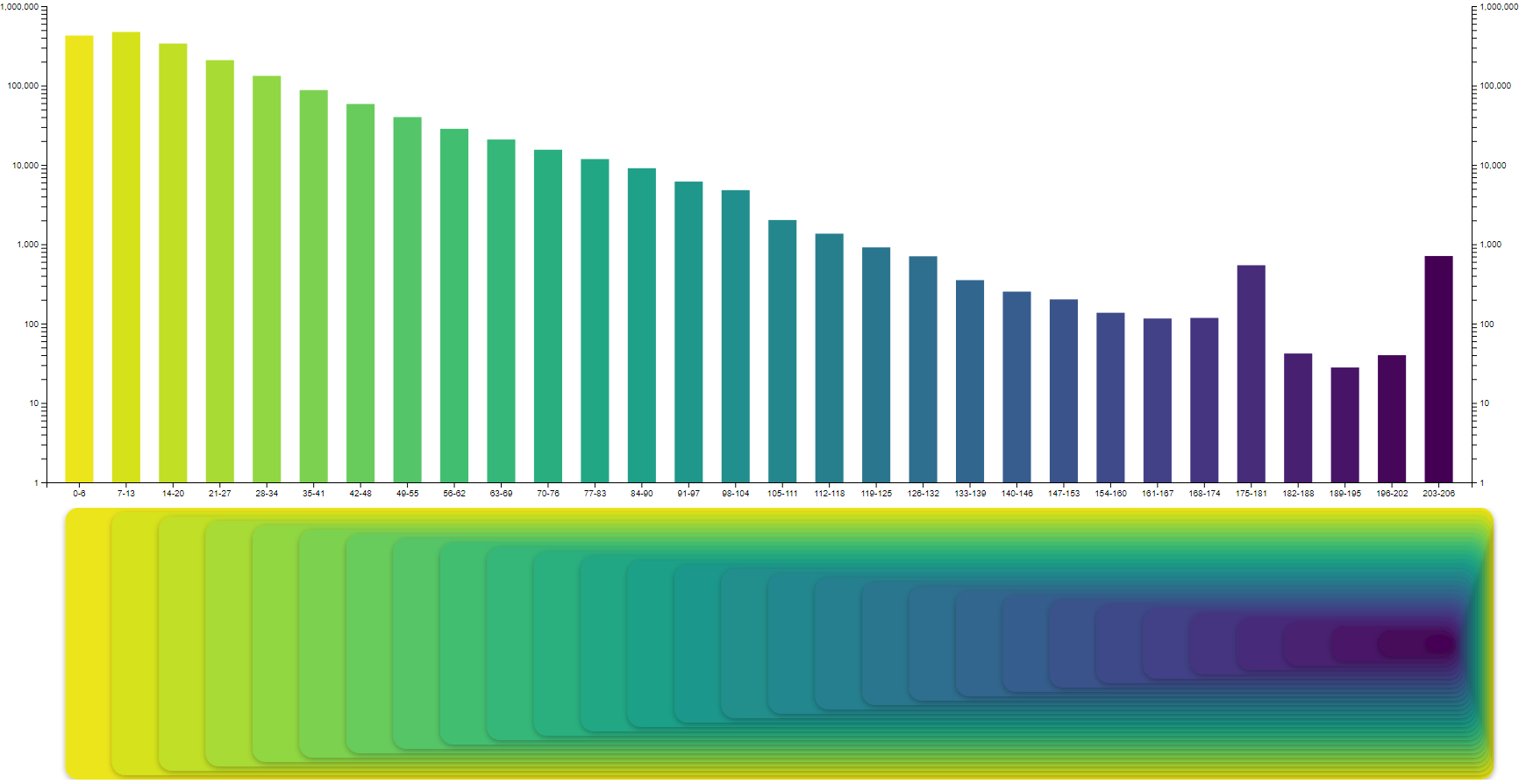}
    \caption{Network wikipedia-growth, coreness scale $=1:7$}
    \label{fig:wikipedia-growth}
\end{figure}

\begin{figure}[tb!]
    \centering
    \includegraphics[width=\textwidth]{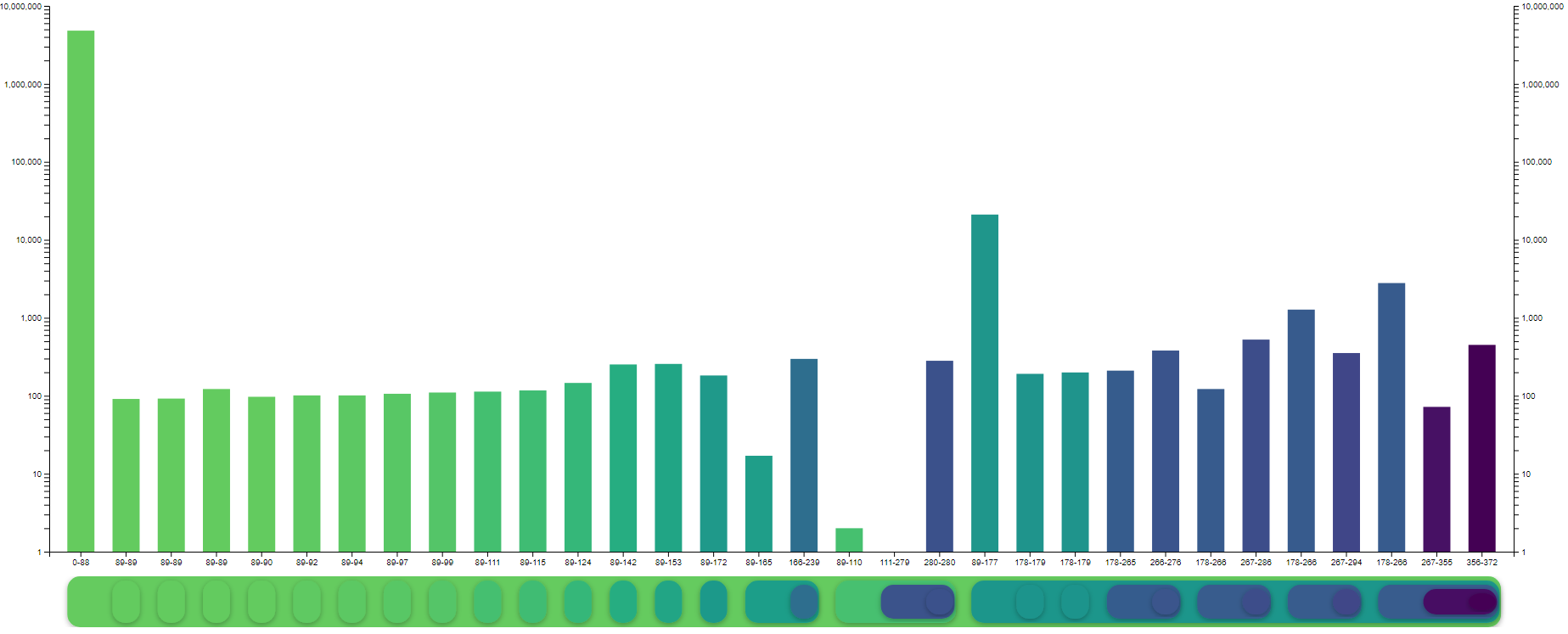}
    \caption{Network soc-LiveJournal1, coreness scale $=1:89$}
    \label{fig:soc-LiveJournal1}
\end{figure}

\begin{figure}[tb!]
    \centering
    \includegraphics[width=\textwidth]{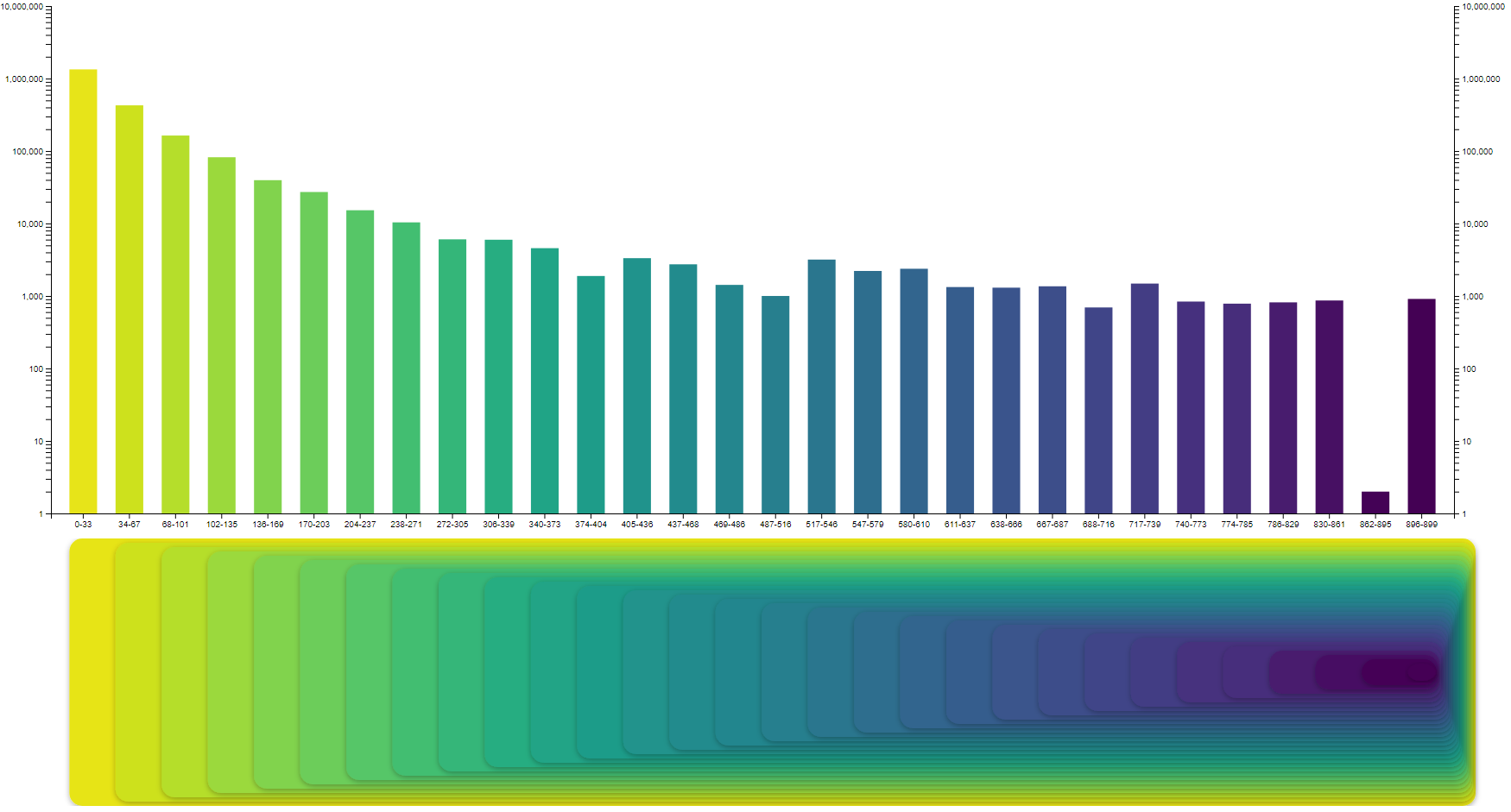}
    \caption{Network wikipedia\_link\_it, coreness scale $=1:34$}
    \label{fig:wikipedia_link_it}
\end{figure}

\begin{figure}[tb!]
    \centering
    \includegraphics[width=\textwidth]{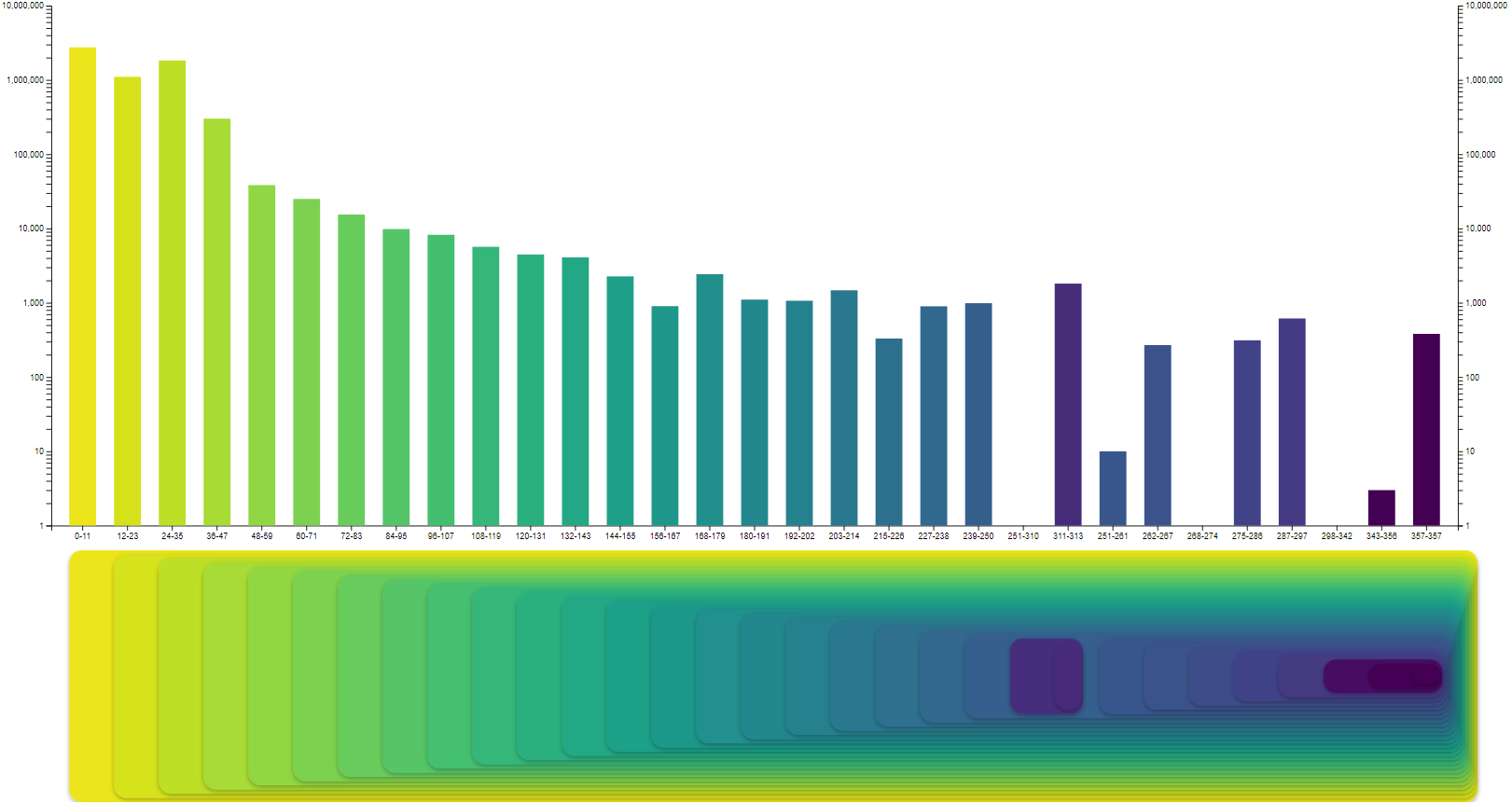}
    \caption{Network wikipedia\_link\_sv, coreness scale $=1:12$}
    \label{fig:wikipedia_link_sv}
\end{figure}

\begin{figure}[tb!]
    \centering
    \includegraphics[width=\textwidth]{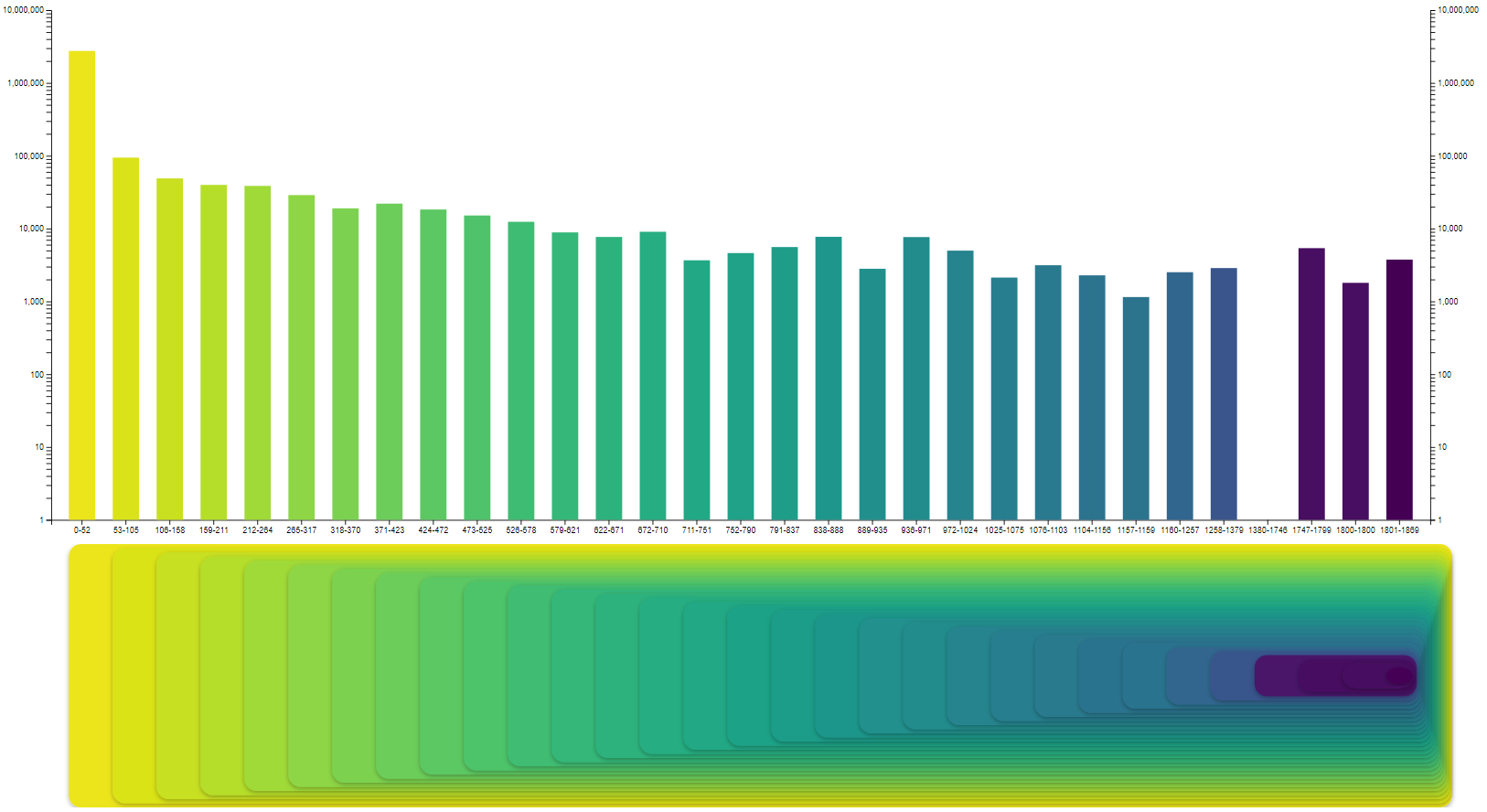}
    \caption{Network wikipedia\_link\_sr, coreness scale $=1:53$}
    \label{fig:wikipedia_link_sr}
\end{figure}
\begin{figure}[tb!]
    \centering
    \includegraphics[width=\textwidth]{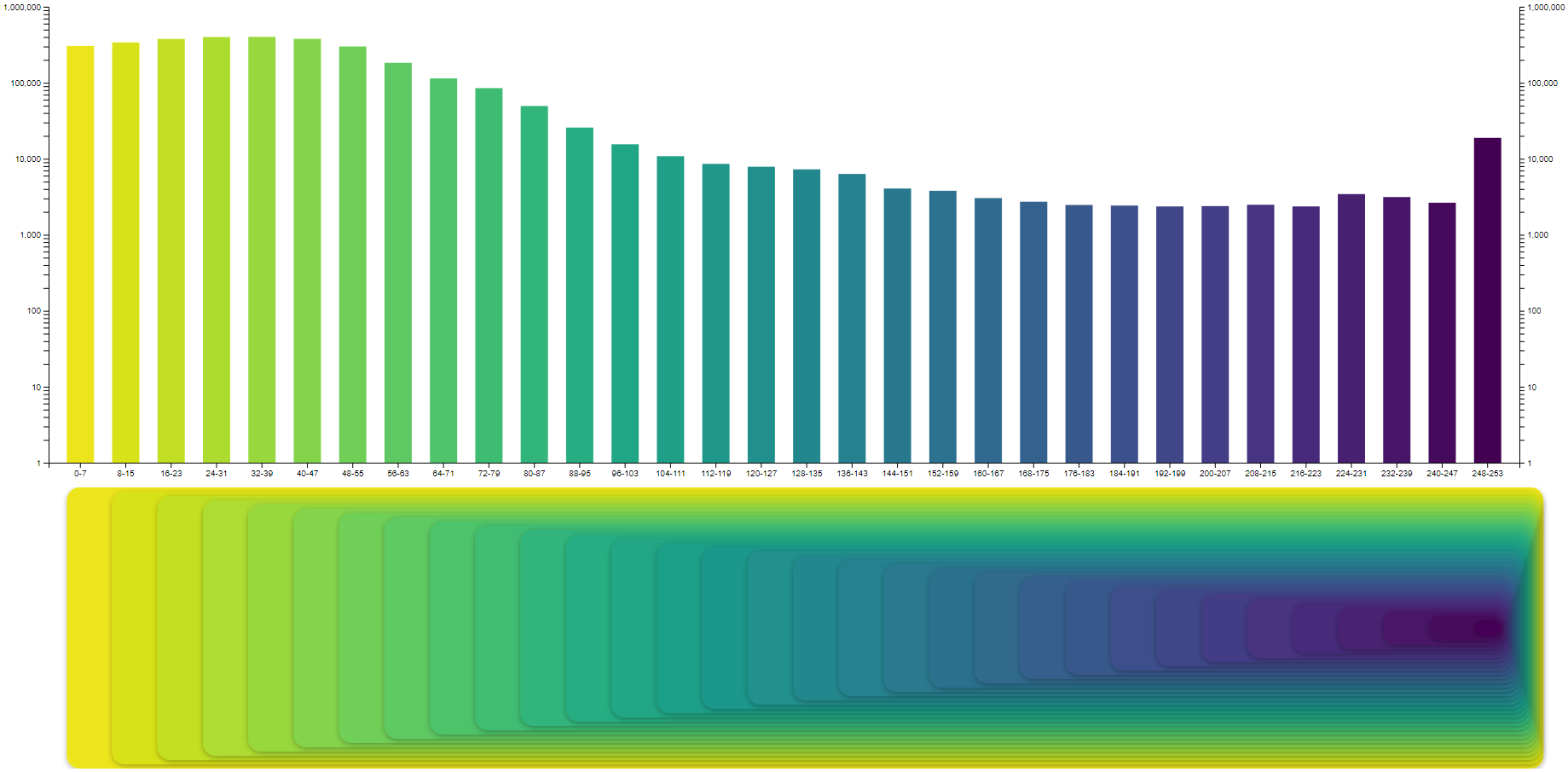}
    \caption{Network orkut-links, coreness scale $=1:8$}
    \label{fig:orkut-links}
\end{figure}
\begin{figure}[tb!]
    \centering
    \includegraphics[width=\textwidth]{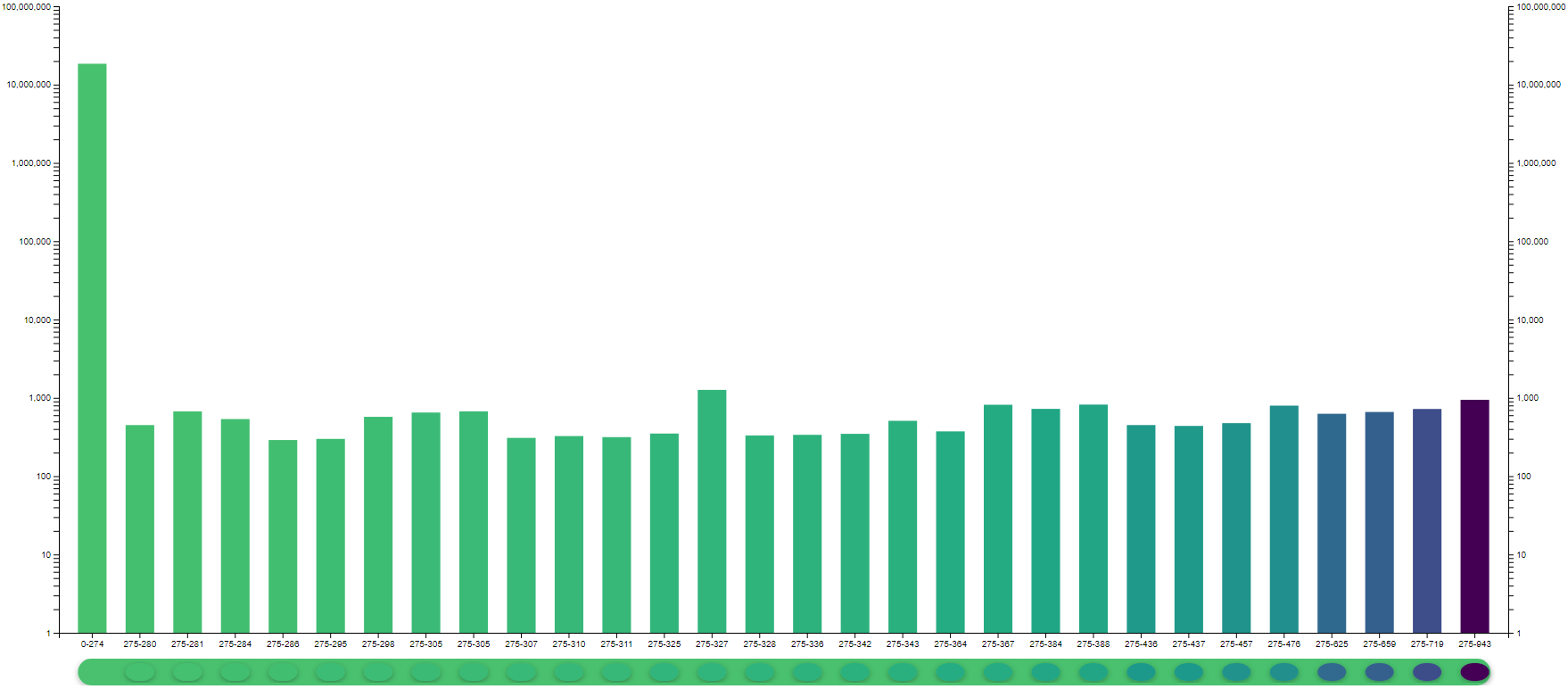}
    \caption{Network dimacs10-uk-2002, coreness scale $=1:275$}
    \label{fig:dimacs10--uk-2002}
\end{figure}
\begin{figure}[tb!]
    \centering
    \includegraphics[width=\textwidth]{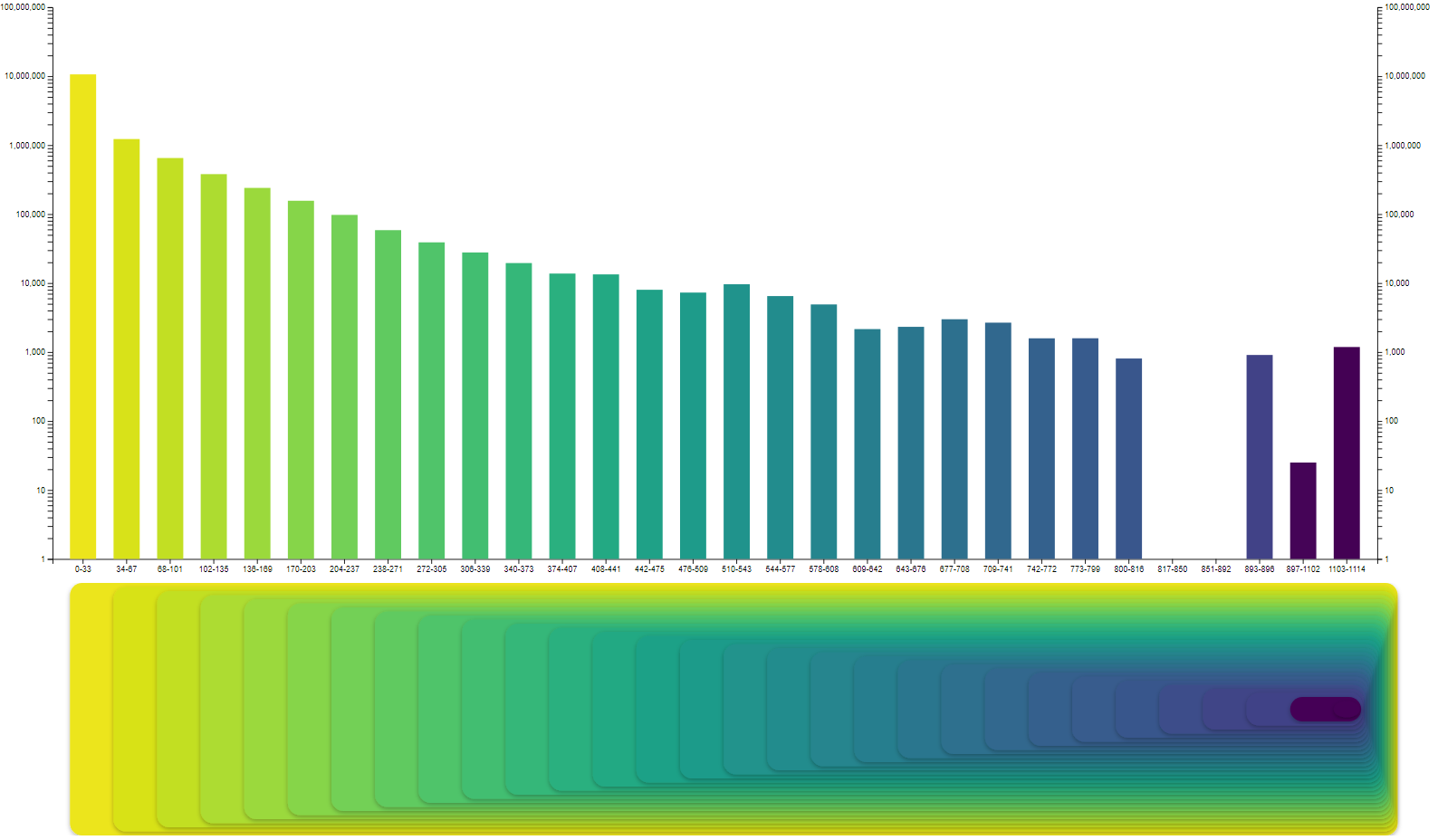}
    \caption{Network wikipedia\_link\_en, coreness scale $=1:34$}
    \label{fig:wikipedia_link_en}
\end{figure}
\begin{figure}[tb!]
    \centering
    \includegraphics[width=\textwidth]{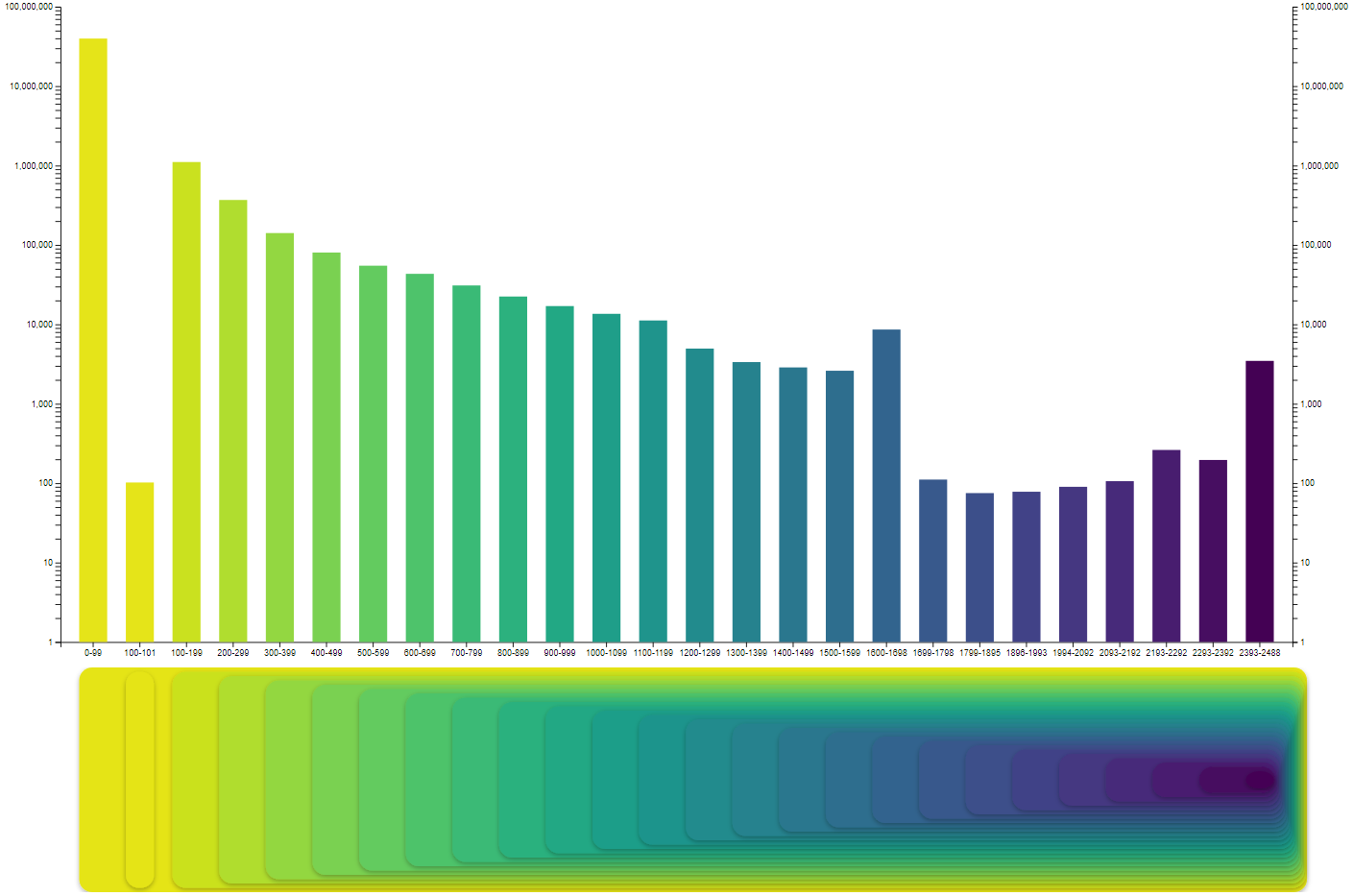}
    \caption{Network twitter, coreness scale $=1:100$}
    \label{fig:twitter}
\end{figure}
\begin{figure}[tb!]
    \centering
    \includegraphics[width=\textwidth]{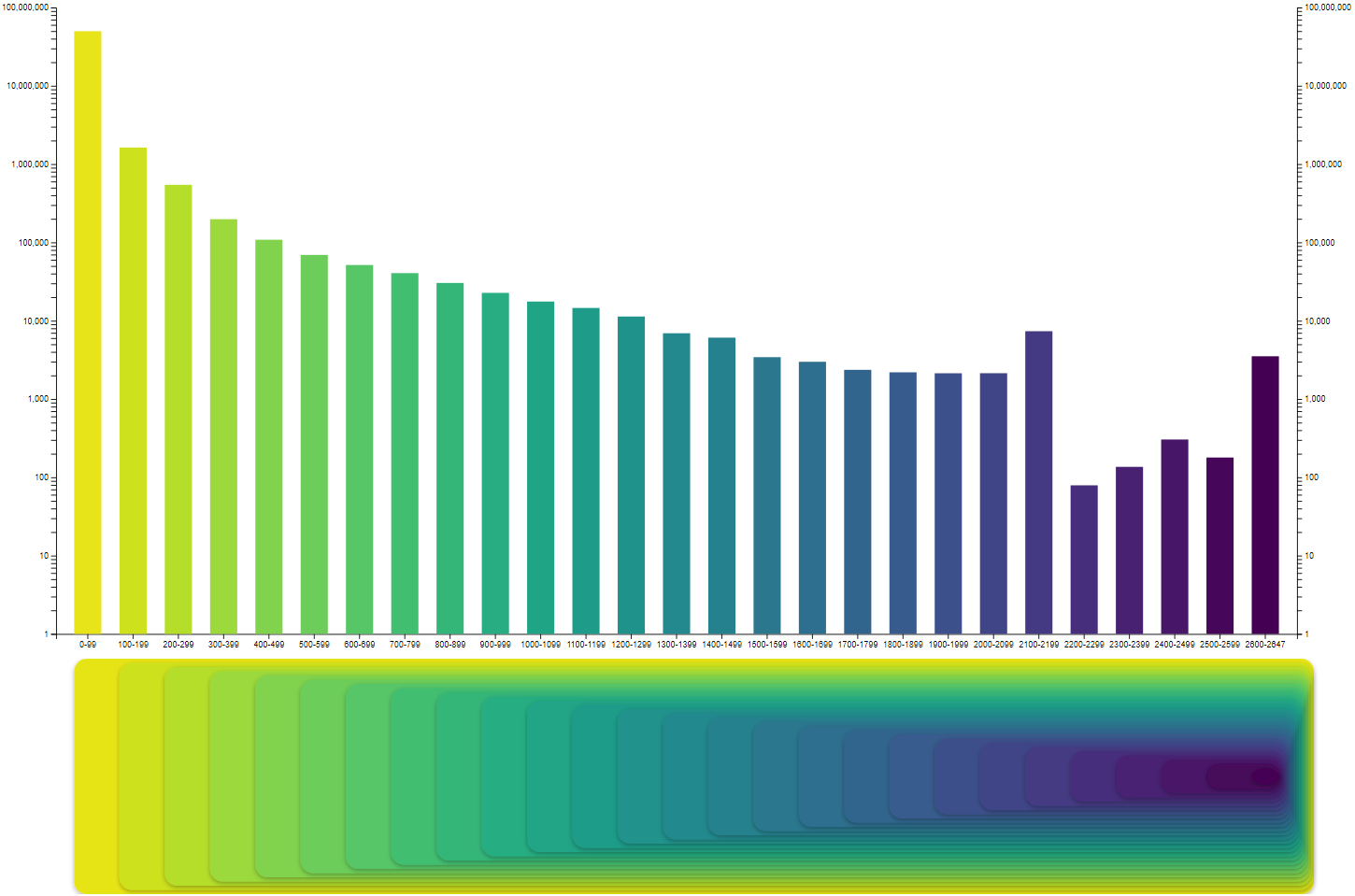}
    \caption{Network twitter\_mpi, coreness scale $=1:100$}
    \label{fig:twitter_mpi}
\end{figure}
\begin{figure}[tb!]
    \centering
    \includegraphics[width=\textwidth]{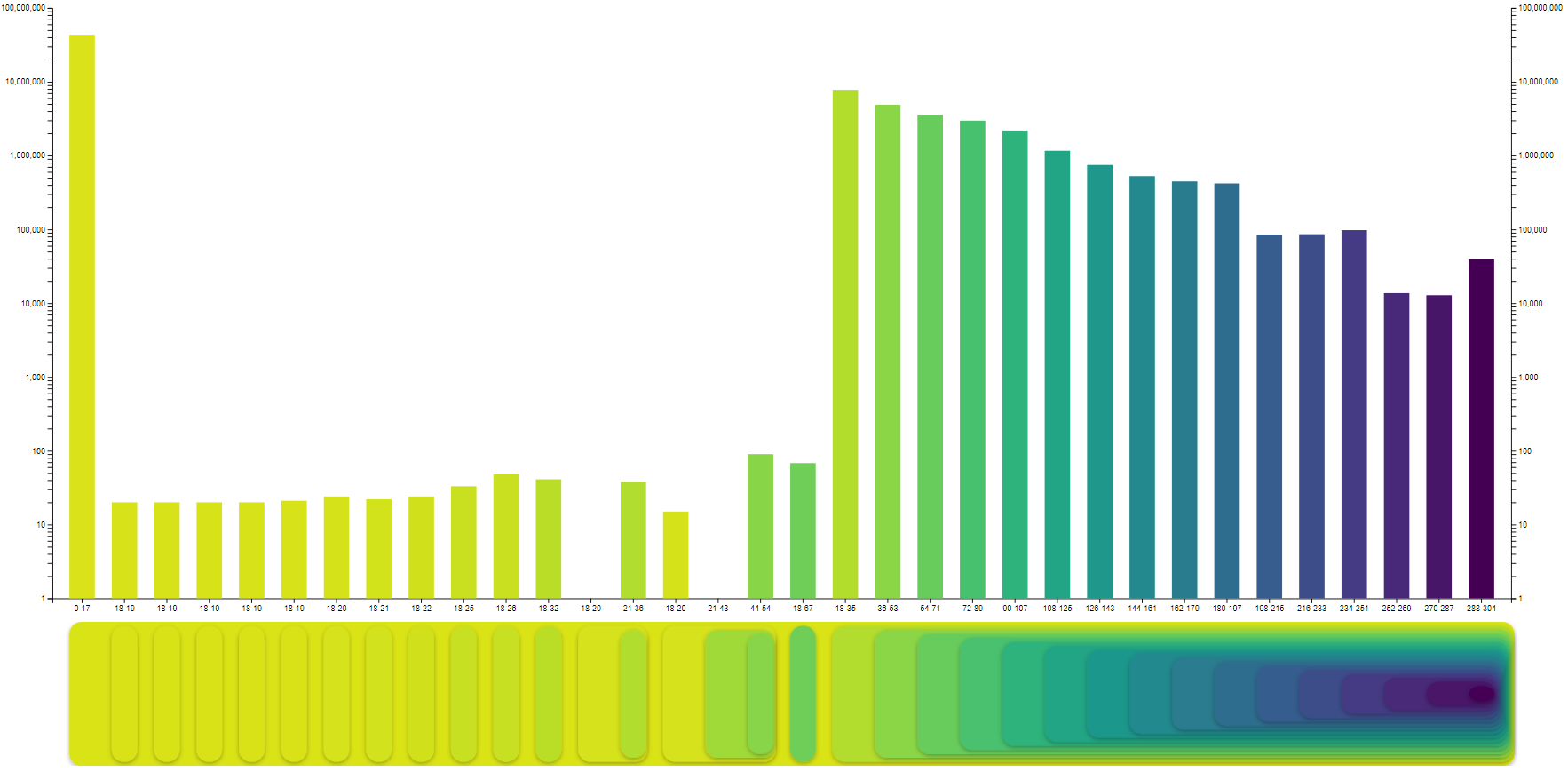}
    \caption{Network friendster, coreness scale $=1:18$}
    \label{fig:friendster}
\end{figure}

\end{document}